\documentclass[3p]{elsarticle}
\usepackage{hyperref}
\usepackage[displaymath, mathlines]{lineno}
\graphicspath{{figures/}}
\usepackage{subcaption}
\usepackage{placeins}
\usepackage{amsmath, amssymb, amsthm}
\usepackage{mathtools}
\usepackage{nicefrac}
\usepackage{relsize}
\modulolinenumbers[1]
\newtheorem{thm}{Theorem}[section]

\newtheorem{rem}[thm]{Remark}
\newtheorem{ass}{Assumption}
\DeclareMathOperator*{\argmin}{argmin}
\newcommand{\cobs}{\mathbf{C}^{\mathrm{obs}}}
\newcommand{\real}[1]{\mathrm{Re} \left( #1 \right) }

\newcommand{\vect}[1]{\mathrm{vec} \left( #1 \right) }
\newcommand{\Cov}[1]{\mathrm{Cov} \left( #1 \right) }
\newcommand{\Var}[1]{\mathrm{Var} \left( #1 \right) }
\newcommand{\norm}[1]{\left\|  #1 \right\|}
\newcommand{\abs}[1]{\left|  #1 \right|}
\newcommand{\vG}{\vect{\mathbf{G}}}
\newcommand{\vGn}{\vect{\mathbf{G}(\mathbf{y}_n)}}
\journal{Journal of Sound and Vibration}

\begin{document}

\begin{frontmatter}

\title{Weighted Data Spaces for Correlation-Based Array Imaging in Experimental Aeroacoustics}

\author[dlraddress]{Hans-Georg Raumer \corref{correspondingauthor}} 
\cortext[correspondingauthor]{Corresponding author}
\ead{hans-georg.raumer@dlr.de}
\author[dlraddress]{Carsten Spehr}
\author[namaddress,mpiaddress]{Thorsten Hohage} 
\author[dlraddress]{Daniel Ernst}
\address[dlraddress]{Institute of Aerodynamics and Flow Technology, German Aerospace Center (DLR), Bunsenstraße 10, 37073 Göttingen}
\address[namaddress]{Institute for Numerical and Applied Mathematics, University of Göttingen, Lotzestraße 16-18, 37073 Göttingen}
\address[mpiaddress]{Max Planck Institute for Solar System Research, Justus-von-Liebig-Weg 3, 37077 Göttingen}

\begin{abstract} 
This article discusses aeroacoustic imaging methods based on correlation measurements in the frequency domain. Standard methods in this field assume that the estimated correlation matrix is superimposed with additive white noise. 
In this paper we present a mathematical model for the measurement process covering arbitrarily correlated noise. 
The covariance matrix of correlation data is given in terms of fourth order moments. The aim of this paper is to explore the use of such additional information on the measurement data in imaging methods.
For this purpose a class of weighted data spaces is introduced, where each data space naturally defines an associated beamforming method with a corresponding point spread function. This generic class of beamformers contains many well-known methods such as Conventional Beamforming, (Robust) Adaptive Beamforming or beamforming with shading. This article examines in particular weightings that depend on the noise (co)variances. In a theoretical analysis we prove that the beamformer, weighted by the full noise covariance matrix, has minimal variance among all beamformers from the described class. Application of the (co)variance weighted methods on synthetic and experimental data show that the resolution of the results is improved and noise effects are reduced.  
\end{abstract}

\end{frontmatter}

%\linenumbers
%%%%%%%%%%%%%%%%%%%%%%%%%%%%%%%%%%%%%%%%%%%%%%%%%%%%%%%%%%%%%
%%%%%%%%%%%%%%%%%% Introduction %%%%%%%%%%%%%%%%%%%%%%%%%%%%%
%%%%%%%%%%%%%%%%%%%%%%%%%%%%%%%%%%%%%%%%%%%%%%%%%%%%%%%%%%%%%

\section{Introduction}
A typical aeroacoustic experiment involves testing a solid object inside the flow field of a wind tunnel and measuring pressure fluctuations with a microphone array apart from that object. For experiments in aeronautical research the solid object may be a rotor blade of a helicopter or a scaled aircraft model, for instance. Such aeroacoustic measurements are conducted in wind tunnels with open or closed test sections. Each environment has its specific challenges and uncertainties regarding the measurement signals. For this paper we will consider a closed test section environment. In such wind tunnels, microphones are usually flush-mounted at the wind tunnel wall and thus located underneath the turbulent boundary layer. Therefore the microphone signals yield a superposition of aeroacoustic source signal, pressure fluctuations due to the turbulent boundary layer and measurement noise.
\ \\ \ \\
Array imaging methods are employed in various fields of physics and engineering such as radar \cite{Haykin1993}, geophysical imaging \cite{Bleistein2001}, speech enhancement \cite{Farmani2017, Ganguly2017, Long2019} or helioseismic holography \cite{Lindsey2000, Lindsey2000a, Gizon2018}. Application of microphone arrays for aeroacoustic purposes was started in 1976 by the work of Billingsley \& Kinns \cite{Billingsley1976} on sound sources of a turbulent jet. Nowadays, in most applications microphone arrays have replaced the elliptic mirrors that were used before (see e.g. \cite{Grosche1975}).
\ \\ \ \\In the aeroacoustic context, array imaging methods are often called beamforming methods and the array image is called source map. 
For any beamforming procedure the essential aspect of data post processing is the extraction of correlation information from the raw microphone signals. The estimated cross correlations then serve as input data for the imaging functional which yields an estimator on the source power of the unknown source at a focus point $\mathbf{y} \in \mathbb{R}^3$. The dynamic range and resolution of the source maps can be improved by optimized experimental setups (see e.g. \cite{Malgoezar2016}) as well as by modifications of the beamforming imaging functionals (see e.g. \cite{Cox1987, Sarradj2006, Dougherty2014}). In this paper we will concentrate on the latter aspect. More precisely, a mathematical framework for the weighting of correlation data is introduced. Within that framework, each weighting type defines an associated beamforming method. The focus is set on weightings depending on the (co)variances of the correlation estimates. This approach is motivated by the technique of \textit{generalized least squares} (GLS) which goes back to Aitken \cite{Aitken1936}. For linear regression models, the \textit{Gauss-Markov Theorem} states that GLS yields the best linear unbiased estimator (BLUE) of the regression coefficient vector. A similar optimality result will be derived in Section \ref{sec:varopt} of this article. More comprehensive and detailed information on GLS can be found e.g. in \cite{Kariya2004}.  \ \\ \ \\
In Section \ref{sec:mathmod} we derive a mathematical model for the measurement process which yields a correlation estimator represented by the sum of the noise free data and an arbitrarily correlated noise quantity. Further we introduce a generic data space that is characterized by a weighting matrix $\mathbf{W}$ that may depend on the covariances of the data and reduces to the standard data space if $\mathbf{W}$ is the identity matrix. Section \ref{sec:arrayapplication} illustrates how this concept applies to imaging methods, in particular beamforming and DAMAS-NNLS. We obtain a class of beamforming and corresponding deblurring methods that are parameterized by the weighting matrix $\mathbf{W}$. Several choices for $\mathbf{W}$ are discussed. On the one hand those which represent well-known beamforming methods (Conventional Beamforming, Robust Adaptive Beamforming, Capon's method, beamforming with shading) and on the other hand those which depend on the variances and covariances of the measurement noise.
A mathematical analysis of the parameterized beamforming methods in Section \ref{sec:varopt} shows that the variance is minimized if $\mathbf{W}$ is chosen as the covariance matrix $\mathbf{\Sigma}$ of the data. Subsequently we point out that this result should not be confused with the variance minimizing property of Capon's method. In particular the Capon beamformer and the beamformer parameterized by $\mathbf{\Sigma}$ are not equivalent. We also demonstrate under which assumptions on the data, Capon's beamformer is a good approximation of the $\mathbf{\Sigma}$-weighted beamformer. In Section \ref{sec:computational} we discuss certain computational aspects regarding the implementation of the methods and the estimation of data covariances. Section \ref{sec:results-synthetic} investigates how the weighting affects the resolution and signal-to-noise ratio (SNR) of the imaging result. The effects are examined for a simple synthetic dataset with a single monopole source. Three particular weighting choices for beamforming and regularized DAMAS-NNLS are tested. In Section \ref{sec:results-experimental} we apply the same set of weightings and imaging methods 
on an experimental dataset and compare their results. Some statistical measures such as the deviation from the white noise assumption are also presented. Section \ref{sec:conclusion} completes this article with some conclusions regarding the theoretical results and application of the presented methods.

%%%%%%%%%%%%%%%%%%%%%%%%%%%%%%%%%%%%%%%%%%%%%%%%%%%%%%%%%%%%%
%%%%%%%%%%%%% Derivation of Data Spaces %%%%%%%%%%%%%%%%%%%%%
%%%%%%%%%%%%%%%%%%%%%%%%%%%%%%%%%%%%%%%%%%%%%%%%%%%%%%%%%%%%%

\section{Measurement Process} \label{sec:mathmod}
The measurement data  of the aeroacoustic experiment is recorded by a microphone array with microphones located at positions $\lbrace \mathbf{x}_1, \dots , \mathbf{x}_M \rbrace$. By $p(\mathbf{x}_m, \omega)$ we denote the Fourier transformed pressure signal at angular frequency $\omega$. For aeroacoustic measurement data it is common practice to treat the observed signals as random variables. Therefore we consider a stochastic framework for the entire analysis. It is assumed that the full signal is represented as the sum of the aeroacoustic source signal $p^{\mathrm{ac}}$,  hydrodynamic pressure fluctuations $p^{\mathrm{hyd}}$ due to the turbulent boundary layer and additive zero mean noise $\epsilon$
\begin{linenomath}\begin{equation}\label{eq:psignal}
p(\mathbf{x}_m, \omega) = p^{\mathrm{ac}} (\mathbf{x}_m, \omega) + p^{\mathrm{hyd}}(\mathbf{x}_m, \omega) + \epsilon(\mathbf{x}_m, \omega) \ .
\end{equation}
\end{linenomath}
For the imaging process we will not use these random variables themselves but only cross correlations of two sensors. For the further modelling of the measurement operator we make the following assumptions.

\begin{ass}[Expectation of Cross Correlations] For the cross correlations in frequency domain holds for $m,l = 1, \dots, M$ \label{as:cor}

\begin{enumerate}
\item \textnormal{All signal parts have zero mean} \label{it:zermean_parts}
\begin{linenomath}
\begin{equation*}
   \mathbb{E}\left[ p^{\textnormal{ac}} (\mathbf{x}_m, \omega)\right] =  \mathbb{E}\left[ p^{\textnormal{hyd}} (\mathbf{x}_m, \omega)\right] = \mathbb{E}\left[ \epsilon (\mathbf{x}_m, \omega)\right] = 0 \ .
\end{equation*}
\item \textnormal{Different signal parts are uncorrelated}
\begin{eqnarray*}
\mathbb{E}\left[ p^{\textnormal{ac}} (\mathbf{x}_m, \omega) p^{\textnormal{hyd}} (\mathbf{x}_l, \omega)^* \right] &= 0 \ ,\\
\mathbb{E}\left[ p^{\textnormal{ac}} (\mathbf{x}_m, \omega) \epsilon ( \mathbf{x}_l,\omega)^* \right] &= 0 \ ,\\
 \mathbb{E}\left[ p^{\textnormal{hyd}} (\mathbf{x}_m, \omega)\epsilon ( \mathbf{x}_l,\omega)^*  \right] &= 0 \ .
\end{eqnarray*}
\item \label{it:blnoise}
\begin{equation*}
\mathbb{E}\left[ p^{\textnormal{hyd}} (\mathbf{x}_m, \omega) p^{\textnormal{hyd}} (\mathbf{x}_l, \omega)^* \right] = 0 \quad \textnormal{if} \  m \neq l \ ,
\end{equation*} 
\item 
\begin{equation*}
\mathbb{E}\left[ \epsilon_m(\omega) \epsilon_l(\omega)^* \right] = 0 \quad \textnormal{if} \ m \neq l \ .
\end{equation*}
\end{linenomath}
\end{enumerate}
\end{ass} 
\begin{rem}[Boundary layer noise]
In a real measurement setup Assumption \textnormal{\ref{as:cor}.\ref{it:blnoise}} is too restrictive, i.e.\ the correlation matrix of hydrodynamic pressure fluctuations will not be diagonal. Nevertheless, the magnitude of diagonal entries (auto powers) will be much greater than the magnitude of off-diagonal entries (cross powers). Investigations on experimental data (turbulent boundary layers on airfoil fuselages respectively on flat plates) show that the magnitude of off-diagonal entries can be well approximated by parametric models decaying exponentially with the distance between the sensors (see e.g. \cite{Efimtsov1982, Palumbo2012, Smolyakov2006}). 
\end{rem}\ \\For the remaining part we omit the frequency dependency to improve readability. During the measuring process $J$ block samples of cross correlations are generated 
\begin{linenomath}
\begin{equation*}
p_j(\mathbf{x}_m) p_j(\mathbf{x}_l)^* \ \text{for} \ j = 1, \dots , J, \ \ m,l = 1, \dots , M  \ .
\end{equation*}
\end{linenomath}
Taking the mean respectively block-sample average yields the observed cross spectral matrix (CSM) $\cobs$
\begin{linenomath}
\begin{equation} \label{eq:csm_def}
\cobs_{ml}(\mathbf{x}_m,\mathbf{x}_l)= \mathbb{M} \left\lbrace p_j(\mathbf{x}_m) p_j(\mathbf{x}_l)^* \right\rbrace = \frac{1}{J}\sum \limits_{j=1}^J  p_j(\mathbf{x}_m, ) p_j(\mathbf{x}_l)^* \ .
\end{equation}
\end{linenomath}
In practice this is usually carried out by Welch's method \cite{Welch1967}. \\ \ \\
The true signal CSM $\mathbf{C}^{\mathrm{ac}}$ is defined as the expectation of the acoustic signal correlations 
\begin{linenomath}
\begin{equation*}
\mathbf{C}^{\mathrm{ac}}_{ml}(\mathbf{x}_m,\mathbf{x}_l) = \mathbb{E}\left[ p^{\textnormal{ac}} (\mathbf{x}_m) p^{\textnormal{ac}} (\mathbf{x}_l)^* \right] \ .
\end{equation*}
\end{linenomath}
Analogously, the diagonal matrices of boundary layer noise and measurement noise correlations are defined by
\begin{linenomath}
\begin{eqnarray*}
    \mathbf{D}^{\mathrm{hyd}}_{ml} &=&  \mathbb{E} \left[ p^{\mathrm{hyd}}(\mathbf{x}_m) p^{\mathrm{hyd}}(\mathbf{x}_l)^*     \right] \\
    \mathbf{D}^{\epsilon}_{ml} &=& \mathbb{E} \left[ \epsilon(\mathbf{x}_m) \epsilon(\mathbf{x}_l)^*     \right] \ .
\end{eqnarray*}
\end{linenomath}
Employing the vector notation $\mathbf{p}_j = \left(p_j(\mathbf{x}_1), \dots, p_j(\mathbf{x}_M) \right)^{\top}$ as well as Eq. \eqref{eq:psignal} and Assumption \ref{as:cor}, we can add and subtract the expected value of each sub-quantity and arrive at
\begin{linenomath}\label{eq:noisemod_der}
\begin{equation} 
\begin{aligned}
\cobs &= \mathbb{M} \left\lbrace \mathbf{p}_j\mathbf{p}_j^* \right\rbrace  \\ &=
\mathbf{C}^{\mathrm{ac}} + \mathbb{M} \left\lbrace \mathbf{p}^{\mathrm{ac}}_j \mathbf{p}^{\text{ac}*}_j \right\rbrace - \mathbf{C}^{\mathrm{ac}} + \mathbf{D}^{\mathrm{hyd}} +  \mathbb{M} \left\lbrace \mathbf{p}^{\mathrm{hyd}}_j  \mathbf{p}^{\mathrm{hyd}*}_j \right\rbrace - \mathbf{D}^{\mathrm{hyd}}  \\
& \  \ \ \ +  \mathbf{D}^{\epsilon} +  \mathbb{M} \left\lbrace \boldsymbol{\epsilon}_j  \boldsymbol{\epsilon}_j^* \right\rbrace - \mathbf{D}^{\epsilon}                                          + \mathbb{M} \left\lbrace \mathbf{p}^{\mathrm{ac}}_j  \mathbf{p}^{\mathrm{hyd}*}_j \right\rbrace  + \mathbb{M} \left\lbrace \mathbf{p}^{\mathrm{hyd}}_j  \mathbf{p}^{\text{ac}*}_j \right\rbrace  \\
& \  \ \ \ +  \mathbb{M} \left\lbrace \mathbf{p}^{\mathrm{ac}}_j  \boldsymbol{\epsilon}_j^* \right\rbrace + \mathbb{M} \left\lbrace \boldsymbol{\epsilon}_j \mathbf{p}^{\mathrm{ac}*}_j \right\rbrace + \mathbb{M} \left\lbrace \mathbf{p}^{\mathrm{hyd}}_j  \boldsymbol{\epsilon}_j^* \right\rbrace + \mathbb{M} \left\lbrace \boldsymbol{\epsilon}_j \mathbf{p}^{\mathrm{hyd}*}_j \right\rbrace  \\ 
&= \mathbf{C}^{\mathrm{ac}} + \mathbf{D} + \mathbf{Z}  \ .
\end{aligned}
\end{equation}
\end{linenomath}
Note that expressions of the type $\mathbb{M}\lbrace \dots \rbrace$ are still random since they are finite sums of random variables, whereas $\mathbb{E}\lbrace \dots \rbrace$ is always deterministic.
Following \eqref{eq:noisemod_der}, the observed CSM can be decomposed as the sum of the true acoustic signal CSM $\mathbf{C}^{\mathrm{ac}}$, a diagonal matrix 
\begin{linenomath}
\begin{equation*}
\mathbf{D}=\mathbf{D}^{\mathrm{hyd}} + \mathbf{D}^{\epsilon} 
\end{equation*}
\end{linenomath}and a noise matrix $\mathbf{Z}$ containing all the remaining summands in \eqref{eq:noisemod_der}. Assumption \ref{as:cor} furthermore implies that 
\begin{linenomath}
\begin{equation*}
\mathbb{E} \left \lbrace \mathbf{Z} \right \rbrace = 0 .
\end{equation*}
\end{linenomath}
 Since we do not want to deal with matrix-valued random variables, we consider their vectorized counterparts $\vect{\cdot}$ instead. For any matrix $\mathbf{A} \in \mathbb{C}^{d \times d}$ we define the column-wise vectorization
 \begin{linenomath}
\begin{equation*}
 \vect{\mathbf{A}} = \left(\mathbf{A}_{11}, \dots , \mathbf{A}_{d 1}, \mathbf{A}_{12} , \dots , \mathbf{A}_{d2}, \dots , \mathbf{A}_{1d} , \dots , \mathbf{A}_{dd} \right)^{\top} \ .
\end{equation*}    
 \end{linenomath}
Hence, $\vect{\mathbf{Z}} \in \mathbb{C}^{M^2}$ is a complex valued random variable with positive semi-definite covariance matrix $\mathbf{\Sigma} \in \mathbb{C}^{M^2 \times M^2}$ 
\begin{linenomath}
\begin{equation*} 
\Cov{\vect{\mathbf{Z} }} = \mathbf{\Sigma}.
\end{equation*} 
\end{linenomath}
Therefore 
\begin{linenomath}
\begin{equation} \label{eq:datadistribution}
\Cov{\vect{(\cobs}} = \Cov{ \vect{\mathbf{C}^{\mathrm{ac}}} + \vect{\mathbf{D}} + \vect{\mathbf{Z}} }  = \mathbf{\Sigma} \ ,
\end{equation} 
\end{linenomath}
since $\mathbf{C}^{\mathrm{ac}}$ and $\mathbf{D}$ are deterministic quantities. Eq. \eqref{eq:datadistribution} reveals that we can estimate $\mathbf{\Sigma}$ by means of the covariances of the components of the observed CSM. In standard aeroacoustic modelling, the CSM is considered as an element of the Hilbert space $\left(\mathbb{C}^{M \times M} , \langle \cdot, \cdot \rangle_F \right) \simeq \left(\mathbb{C}^{M^2} , \langle \cdot, \cdot \rangle_2 \right)$ with inner product
\begin{linenomath}
\begin{equation*}
\left \langle \mathbf{A}, \mathbf{B}  \right \rangle_F = \sum \limits_{m,l=1}^M \mathbf{A}_{ml} \left(\mathbf{B}_{ml} \right)^*  =
 \sum \limits_{j=1}^{M^2} \vect{\mathbf{A}}_j  \left( \vect{\mathbf{B}}_{j} \right)^* = \left \langle \vect{\mathbf{A}}, \vect{\mathbf{B}}  \right \rangle_2 \ .
\end{equation*} 
\end{linenomath}
For any Hermitian, positive definite matrix $\mathbf{W}$ we can instead consider the Hilbert space $\left(\mathbb{C}^{M^2} , \langle \cdot, \cdot \rangle_{\mathbf{W}} \right)$ with inner product
\begin{linenomath}
\begin{equation}\label{eq:innersig}
\begin{aligned}
\left \langle \vect{\mathbf{A}}, \vect{\mathbf{B}}  \right \rangle_{\mathbf{W}} &= \left \langle \mathbf{W}^{\nicefrac{-1}{2}} \vect{\mathbf{A}}, \mathbf{W}^{\nicefrac{-1}{2}} \vect{\mathbf{B}}  \right \rangle_{2}  \\ &= \left \langle  \vect{\mathbf{A}}, \mathbf{W}^{-1} \vect{\mathbf{B}}  \right \rangle_{2} , 
\end{aligned}
\end{equation}
\end{linenomath}
where $\mathbf{W}^{-1}$ denotes the inverse matrix. \footnote{
Vice versa, any Hilbert space $\mathcal{H}$ on $\mathbb{C}^{M^2}$ can be characterized as $\mathcal{H} = \left(\mathbb{C}^{M^2}, \langle \cdot , \cdot \rangle_{\mathbf{W}}  \right)$ with an appropriate Hermitian and positive definite matrix $\mathbf{W}$ called Gramian matrix (see e.g. \cite{Horn2012}).} Let $\mathbf{U} \mathbf{\Lambda} \mathbf{U}^*$ be the eigenvalue decomposition (EVD) of $\mathbf{W} $ with a diagonal matrix $\mathbf{\Lambda}$  and a unitary matrix $\mathbf{U}$ then $\mathbf{W}^{-\nicefrac{1}{2}}$ is defined by
\begin{linenomath}
\begin{equation*}
\mathbf{W}^{-\nicefrac{1}{2}} = \mathbf{U} \mathbf{\Lambda}^{-\nicefrac{1}{2}} \mathbf{U}^* \ .
\end{equation*} 
\end{linenomath}
We recall that standard beamforming methods can be characterized by a minimization problem (see e.g. \cite{Sijtsma2004}). The purpose of  introducing the above concept is to formulate such a characterizing minimization problem for any Hilbert space $\left(\mathbb{C}^{M^2} , \langle \cdot, \cdot \rangle_{\mathbf{W}} \right)$. This will be done in the next section.

%%%%%%%%%%%%%%%%%%%%%%%%%%%%%%%%%%%%%%%%%%%%%%%%%%%%%%%%%%%%%
%%%%%%% Application to Array Imaging Methods %%%%%%%%%%%%%%%%
%%%%%%%%%%%%%%%%%%%%%%%%%%%%%%%%%%%%%%%%%%%%%%%%%%%%%%%%%%%%%

\section{Application to Array Imaging Methods} \label{sec:arrayapplication}
So far we presented a general mathematical modelling of the measurement process and motivated the choice of a generic data space $\left(\mathbb{C}^{M^2} , \langle \cdot, \cdot \rangle_{\mathbf{W}} \right)$. As a next step we wish to incorporate this concept into aeroacoustic source localization methods. This will be done for beamforming and DAMAS-NNLS.

\subsection{Sound propagation}
For measurements in closed test sections, usually a simplified free field sound propagation model is applied (see e.g. \cite{Sijtsma2012}). Let $c$ denote the speed of sound and consider a homogeneous flow field $\mathbf{u} = (u_1, u_2, u_3)^{\top}$ with $\abs{\mathbf{u}}<c$ (subsonic flow). With the time factor convention $e^{+ i \omega t}$, time harmonic sound propagation in this flow field is modelled by the convected Helmholtz equation (see e.g. \cite{Ostashev2015}) %eqn. 2.88 p.50 
\begin{linenomath}
\begin{equation}
(k - \mathrm{i} \mathbf{m} \cdot \nabla)^2p + \Delta p =  -Q \ .
\label{eq:convhelmeq}
\end{equation}
\end{linenomath}
For a source term $Q$, wavenumber $k = \frac{\omega}{c}$ and Mach vector $\mathbf{m} = \frac{1}{c} \mathbf{u}$. The Green's function $g(\mathbf{x}, \mathbf{y}, \omega)$  of Eq. \eqref{eq:convhelmeq} is (see \cite[Appendix A]{Mosher1984})
\begin{linenomath}
\begin{equation*} 
g(\mathbf{x},\mathbf{y}, \omega) =  \frac{\textnormal{exp} \left(\frac{-\mathrm{i} k}{\beta^2} \left( -(\mathbf{x}-\mathbf{y})\cdot \mathbf{m} + |\mathbf{x}-\mathbf{y}|_{\mathbf{m}} \right) \right) }{4 \pi |\mathbf{x}-\mathbf{y}|_{\mathbf{m}} } \ ,
\end{equation*}
\end{linenomath}
with $\beta^2 = 1-\abs{\mathbf{m}}^2$ and the distance measure 
\begin{linenomath}
\begin{equation*}
|\mathbf{x}-\mathbf{y}|_{\mathbf{m}} = \sqrt{\left((\mathbf{x}-\mathbf{y}) \cdot \mathbf{m} \right)^2 + \beta^2 |\mathbf{x}-\mathbf{y}|^2} \ .
\end{equation*}
\end{linenomath}

\subsection{Beamforming} 
We start with a brief introduction to beamforming methods. For a source region of interest $\mathcal{Y}$, the beamforming imaging functional is a map
\begin{linenomath}
\begin{equation*}
\mathcal{I}(\cdot, \omega): \ \mathcal{Y} \rightarrow \mathbb{C} , \ \ \ \mathbf{y} \mapsto \mathcal{I}(\mathbf{y}) \ .
\end{equation*}
\end{linenomath}
In the following we will omit the dependency on $\omega$ to increase the readability.  For a potential source location $\mathbf{y}$ we define the \textit{propagation vector}
\begin{linenomath}
\begin{equation*}
\mathbf{g}(\mathbf{y}) =
\begin{pmatrix}
g(\mathbf{x}_1, \mathbf{y}) \\
\vdots \\
g(\mathbf{x}_M, \mathbf{y}) 
\end{pmatrix}
\end{equation*}
\end{linenomath}
and the corresponding \textit{propagation matrix}
\begin{linenomath}
\begin{equation*}
\mathbf{G}(\mathbf{y})  = \mathbf{g}(\mathbf{y})  \mathbf{g}(\mathbf{y})^* \ .
\end{equation*}\end{linenomath}
The principle of frequency domain beamforming is to minimize the norm distance between a scalar multiple of the propagation matrix at the source point $\mathbf{y}$ and the measured CSM  i.e.
\begin{linenomath}\begin{equation}\label{eq:minprob}
\mathcal{I}(\mathbf{y}) = \argmin \limits_{\mu \in \mathbb{C}} \norm{\cobs- \mu \mathbf{G}}^2 \ .
\end{equation}\end{linenomath}
Standard beamforming methods consider the Frobenius norm as distance measure i.e. the Euclidean norm of the vectorized quantity. 
Solving Eq. \eqref{eq:minprob} with respect to the Frobenius norm is often referred to as \textit{Conventional Beamforming} \cite{Sijtsma2004}. 
\\ \ \\Alternatively one may choose the norm $\norm{\cdot}_{\mathbf{W} }$ induced by the scalar product $\langle \cdot , \cdot \rangle_{\mathbf{W}}$ \eqref{eq:innersig}
\begin{linenomath}\begin{equation*}
\norm{\mathbf{X}}_{\mathbf{W}} = \sqrt{\left \langle  \mathbf{X},  \mathbf{X} \right \rangle_{\mathbf{W}} } \ . 
\end{equation*}\end{linenomath}
The minimization problem 
\begin{linenomath}\begin{equation} \label{eq:minprob_wbf}
\mathcal{I}_{\mathbf{W}}(\mathbf{y}) = \argmin \limits_{\mu \in \mathbb{C}} \norm{ \vect{\cobs}- \mu \vG }_{\mathbf{W}}^2 
\end{equation}\end{linenomath}
has the solution
\begin{linenomath}\begin{equation} \label{eq:wbf_def}
\mathcal{I}_{\mathbf{W}}(\mathbf{y}) =   \frac{ \left \langle \vect{\cobs}, \vG \right \rangle_{\mathbf{W}} } {\left \langle \vG , \vG  \right \rangle_{\mathbf{W}}} .
\end{equation}\end{linenomath}
\begin{rem}[Real valued and non-negative source power] \ \label{rem:realnonneg_sp} \\
The minimization problem \eqref{eq:wbf_def} is formulated over $\mathbb{C}$ since this is the natural choice for complex valued vector spaces. Note that the minimizer over $\mathbb{R}$ is simply given by the real part of the complex solution. For any complex valued beamforming solution provided by \eqref{eq:wbf_def} one may then consider 
\begin{linenomath}\begin{equation*}
\mathrm{max} \left(0, \real{\mathcal{I}_{\mathbf{W}} } \right)
\end{equation*}\end{linenomath}
as an estimator of the source power. 
\end{rem}

%%%%%%%%%%%%%%%%%%%%%% Einordnung %%%%%%%%%%%%%%%%%%%%%%%%%%%%%%%%%%%%%%%

\subsubsection{Choice of the weighting matrix}
The presented generic beamformer covers arbitrary Hermitian and positive definite weighting matrices $\mathbf{W}$. Now we will examine several particular choices of $\mathbf{W}$ following two intentions:
\begin{enumerate}
\item Introduction of weighting choices related to the data covariance matrix $\mathbf{\Sigma}$.
\item Incorporation of common imaging methods (Conventional Beamforming, Robust Adaptive Beamforming, Capon's method, beamforming with shading) into the framework of Section \ref{sec:mathmod}.
\end{enumerate}

\paragraph{Conventional Beamforming:} \ \\
\begin{linenomath}\begin{equation}
\mathbf{W} = \sigma^2 \mathbf{I} \ \  \mathrm{for} \ \sigma^2 > 0 \ . \label{eq:wc_conventional}
\end{equation}\end{linenomath}
If we choose the weighting matrix $\mathbf{W}$ as a positive multiple of the identity matrix, the induced norm coincides with the Frobenius norm up to the factor $\sigma^2$. Hence the minimization problem \eqref{eq:minprob} is equivalent to the minimization with respect to the Frobenius norm
\begin{linenomath}\begin{equation*}
\mathcal{I}_{\mathbf{I}}(\mathbf{y}, \omega) = \argmin \limits_{\mu \in \mathbb{C}} \norm{\cobs- \mu \mathbf{G}}_F^2 \ ,
\end{equation*}\end{linenomath}
which yields the Conventional Beamforming solution.

\paragraph{Diagonal inverse covariance weighting \textnormal{(iv-d)}:} \ \\
\begin{linenomath}\begin{equation}
\mathbf{W}_{ij} = \begin{cases} \mathbf{\Sigma}_{ij} & \text{if} \  i=j \\ 0 & \text{else} \end{cases} \ .
 \label{eq:wc_iv-d}
\end{equation}\end{linenomath}
This weighting matrix contains the variances of each CSM entry 
\begin{linenomath}\begin{equation*}
\sigma_{ml}^2 =  \Var{ \cobs_{ml} } \ .
\end{equation*}\end{linenomath}
Since all off-diagonal entries of $\mathbf{W}$ vanish by construction, the corresponding norm can be rewritten in matrix notation
\begin{linenomath}\begin{equation*}
\norm{\mathbf{A}}_{\mathbf{W}}^2 = \sum \limits_{m=1}^M \sum \limits_{l=1}^M \frac{1}{\sigma_{ml}^2} \abs{\mathbf{A}_{ml}}^2 \ ,
\end{equation*}\end{linenomath}
which is a weighted Frobenius norm with the inverse variances as weights. This can be regarded as a reliability criterion on the data i.e. datapoints with high variances are considered less reliable and get a low weight in the optimization process whereas datapoints with low variances get a higher weight.

\paragraph{Full inverse covariance weighting \textnormal{(iv-f)}:} \ \\
\begin{linenomath}\begin{equation}
\mathbf{W} = \mathbf{\Sigma}  \label{eq:wc_iv-f} \ .
\end{equation}\end{linenomath}
The norm distance induced by $\norm{\cdot}_{\mathbf{\Sigma}}$ is also known as Mahalanobis distance \cite{Mahalanobis1936}. Furthermore this weighting choice performs a whitening transformation which is specified in the following remark. 
\begin{rem}[Whitening] \label{rem:whitening}\ \\
The choice \eqref{eq:wc_iv-f} of the weighting matrix $\mathbf{W}$ can be regarded as a whitening transformation that results in uncorrelated data with the identity matrix as covariance matrix i.e. Mahalanobis whitening. Similarly the choices \eqref{eq:wc_iv-d} and \eqref{eq:wc_conventional} are also whitening transformations under the additional assumption that the noise is uncorrelated (case \eqref{eq:wc_iv-d}) or even  white (case  \eqref{eq:wc_conventional}). 
\end{rem}

\paragraph{Beamforming with shading:} \ \\
To improve the imaging result one may endow each microphone with a weighting factor $\nu_m>0$ (see e.g. \cite{Sijtsma2010, Brooks1987}). This procedure is often called \textit{shading} in the literature  and yields the imaging functional
\begin{linenomath}\begin{equation*}
\mathcal{I}^{\mathrm{shad}} = \argmin_{\mu \in \mathbb{C}} \sum \limits_{m=1}^M \sum \limits_{l=1}^M \abs{\cobs_{ml} - \mu \nu_m \nu_l \mathbf{g}_m \mathbf{g}_l^* }^2 \ .
\end{equation*}\end{linenomath}
Hence, in the framework of this article, the beamforming functional with shading is represented by the weighting matrix 
\begin{linenomath}\begin{equation*}
\mathbf{W} = \mathrm{diag}\left( \vect{\boldsymbol{\nu} \boldsymbol{\nu}^{\top}} \right)^{-1} \ . 
\end{equation*}\end{linenomath}
Where $\boldsymbol{\nu} = \left(\nu_1, \dots , \nu_M   \right)^{\top}$ denotes the vector of microphone weights.
\paragraph{Robust Adaptive Beamforming (RAB):} \ \\
This beamforming method was introduced by Cox et al. \cite{Cox1987} and depends on a modelling parameter $\alpha>0$. It is defined as  
\begin{linenomath}\begin{equation*}
\begin{aligned}
\mathcal{I}^{\mathrm{RAB}} &= \frac{\mathbf{g}^* \left( \cobs + \alpha \mathbf{I} \right)^{-1} \cobs \left( \cobs + \alpha \mathbf{I} \right)^{-1} \mathbf{g}}{\left(  \mathbf{g}^* \left( \cobs + \alpha \mathbf{I} \right)^{-1} \mathbf{g} \right)^2} \\
&= \frac{\left \langle \vect{\mathbf{R}^{-1} \cobs \mathbf{R}^{-1}} , \vG \right \rangle_2}{ \left \langle \vect{\mathbf{R}^{-1} \mathbf{G} \mathbf{R}^{-1}} , \vG \right \rangle_2 }
\ ,
\end{aligned}
\end{equation*}\end{linenomath}
where $\mathbf{R} = \cobs + \alpha \mathbf{I}$. Again this beamforming method can be represented within the framework of this article by the weighting matrix
\begin{linenomath}\begin{equation}
\mathbf{W} = \mathbf{R}^{\top} \otimes \mathbf{R} \ ,   \label{eq:wc_rab}
\end{equation}\end{linenomath}
where $\otimes$ denotes the Kronecker product \cite[Def. 4.2.1 p.243]{Horn1991}. This statement can be verified using elementary properties of the Kronecker product. Firstly we note that $\mathbf{W}$ in \eqref{eq:wc_rab} is regular/Hermitian/positive definite if and only if $\mathbf{R}$ is regular/Hermitian/positive definite (see \cite[p. 243 ff.]{Horn1991}). From \cite[Lemma 4.3.1 p.255]{Horn1991} we obtain further that for any $\mathbf{A} \in \mathbb{C}^{M \times M}$ holds 
\begin{linenomath}\begin{equation}
\begin{aligned}
\left(\mathbf{W}\right)^{-1} \vect{\mathbf{A}} &= \left( \mathbf{R}^{-\top} \otimes \mathbf{R}^{-1} \right) \vect{\mathbf{A}} \\
&= \vect{\mathbf{R}^{-1} \mathbf{A} \mathbf{R}^{-1} } \ .
\end{aligned} \label{eq:matrixeq_kron}
\end{equation}\end{linenomath}
Using Eq. \eqref{eq:matrixeq_kron} multiple times for $\mathbf{A}=\cobs$ resp. $\mathbf{A}=\mathbf{G}$ yields
\begin{linenomath}\begin{equation*}
\begin{aligned}
\mathcal{I}^{\mathrm{RAB}} &= \frac{\left \langle \vect{\mathbf{R}^{-1} \cobs \mathbf{R}^{-1}} , \vG \right \rangle_2}{ \left \langle \vect{\mathbf{R}^{-1} \mathbf{G} \mathbf{R}^{-1}} , \vG \right \rangle_2 } \\
&=\frac{\left \langle \mathbf{W}^{-1} \vect{ \cobs } , \vG \right \rangle_2}{ \left \langle \mathbf{W}^{-1} \vG , \vG \right \rangle_2 } \ .
\end{aligned}
\end{equation*}\end{linenomath}
\paragraph{Capon's method:} \ \\
This method was introduced by Capon \cite{Capon1969} and can be regarded as the limit case $\alpha=0$ of RAB. Assuming that $\cobs$ is regular, it is defined as
\begin{linenomath}\begin{equation}
\mathcal{I}^{\mathrm{Cap}} =    \left( \mathbf{g}^*\left(\cobs\right)^{-1} \mathbf{g} \right)^{-1} \ . \label{eq:caponbf_def}
\end{equation}\end{linenomath}
Analogously to RAB, Capon's method is represented by the weighting matrix
\begin{linenomath}\begin{equation}
\mathbf{W} = \left(\cobs\right)^{\top} \otimes \cobs \ .         \label{eq:wc_capon}
\end{equation}\end{linenomath}

\subsection{DAMAS}
The DAMAS problem was introduced by Brooks \& Humphreys in 2006 and attempts to deblur beamforming source maps, by solving an inverse problem. DAMAS relates the unknown source data to the beamforming result by means of an integral kernel, the \textit{point spread function} (PSF). For a map region $\mathcal{Y} \subset \mathbb{R}^{3}$ and a fixed frequency $\omega$ this integral kernel is defined by
\begin{linenomath}\begin{equation*}
\psi: \ \mathcal{Y} \times \mathcal{Y} \rightarrow \mathbb{C}, \ \ \ \left(\mathbf{y}, \mathbf{y}' \right) \mapsto \argmin \limits_{\mu \in \mathbb{C}} \norm{ \vect{\mathbf{G}(\mathbf{y}')}- \mu \vect{\mathbf{G}(\mathbf{y})} }_{2}^2 \ .
\end{equation*}\end{linenomath}
Note that the point spread function yields the value of the beamforming functional at focus point $\mathbf{y}$ for a monopole source at $\mathbf{y}'$. For an acoustic source map $\mathcal{I}$ obtained by Conventional Beamforming, the standard DAMAS problem seeks a pointwise positive solution $q$ of the Fredholm integral equation of the first kind
\begin{linenomath}\begin{equation} \label{eq:damas_inteq}
\mathcal{I}(\mathbf{y}) = \int_{\mathcal{Y}} q(\mathbf{y}') \psi(\mathbf{y}, \mathbf{y}') d\mathbf{y}' \ .
\end{equation}\end{linenomath}
If the PSF is shift invariant, i.e. $\psi(\mathbf{y}, \mathbf{y}') = \widetilde{\psi}(\mathbf{y}-\mathbf{y}')$ the integral reduces to a convolution integral. For aeroacoustic measurement setups the PSF is usually not shift invariant, nevertheless methods to solve integral equations as \eqref{eq:damas_inteq} are often referred to as \textit{deconvolution methods} in the aeroacoustic community. For a finite set of focus points $ \lbrace \mathbf{y}_n  \rbrace_{n=1}^N \subset \mathcal{Y}$ the integral Eq. \eqref{eq:damas_inteq} may be discretized as 
\begin{linenomath}\begin{equation}
\mathbf{H} \mathbf{q}  = \mathbf{b} \ , \label{eq:damas_inteq_discrete} 
\end{equation}\end{linenomath}
with 
\begin{linenomath}\begin{equation*}
\mathbf{H}_{nl} = \psi(\mathbf{y}_l, \mathbf{y}_n) \quad  \mathrm{and} \quad \mathbf{b}_n = \mathcal{I}(\mathbf{y}_n) \quad \mathrm{for} \ n,l=1, \dots , N  .
\end{equation*}\end{linenomath}
In the original paper by Brooks \& Humphreys \cite{Brooks2006}, the discrete problem \eqref{eq:damas_inteq_discrete} is solved by a Gauss-Seidel Algorithm applied to the unconstrained problem and the non-negativity constraint is enforced after each iteration. However, the Gauss-Seidel approach may lead to unsatisfactory results for experimental data sets. Experimental investigations in \cite{Bahr2017} have shown that solving the non-negative least squares problem (DAMAS-NNLS) 
\begin{linenomath}\begin{equation}
\min \limits_{\mathbf{q} \geq 0} \norm{\mathbf{Hq}-\mathbf{b}}_2^2 \label{eq:damas_nnls} 
\end{equation}\end{linenomath} 
yields cleaner source maps than the original DAMAS approach. Problem \eqref{eq:damas_nnls} may be solved by an appropriate solver such as L-BFGS-B \cite{Zhu1997} or an active set method  \cite[p. 161]{Lawson1974}. DAMAS-NNLS does not guarantee a unique solution and therefore the results of different solver routines may differ. Note that the minimization problem \eqref{eq:damas_nnls} is formulated on the source space and not on the data space as the minimization problem \eqref{eq:minprob_wbf} that characterizes beamforming results. 
\ \\ \ \\
For a source map obtained by a minimization problem with respect to a weighted norm $\norm{\cdot }_{\mathbf{W}}$ we just have to choose the same norm in the definition of the PSF i.e.
\begin{linenomath}\begin{equation*}
\psi_{\mathbf{W}}: \ \mathcal{Y} \times \mathcal{Y} \rightarrow \mathbb{C}, \ \ \ \left(\mathbf{y}, \mathbf{y}' \right) \mapsto \argmin \limits_{\mu \in \mathbb{C}} \norm{ \vect{\mathbf{G}(\mathbf{y}')} - \mu \vect{\mathbf{G}(\mathbf{y})} }_{\mathbf{W}}^2 \ .
\end{equation*}\end{linenomath} 
Finally we solve the NNLS-problem 
\begin{linenomath}\begin{equation}
\min \limits_{\mathbf{q} \geq 0} \norm{\mathbf{H}_{\mathbf{W}}\mathbf{q}-\mathbf{b}_{\mathbf{W}}}_2^2 \label{eq:damas_nnlsw} 
\end{equation}\end{linenomath}
with 
\begin{linenomath}\begin{equation*}
\left[{\mathbf{{H}}_{\mathbf{W}}}\right]_{nl} = \psi_{\mathbf{W}}(\mathbf{y}_l, \mathbf{y}_n) \quad  \mathrm{and} \quad \left[{\mathbf{b}_{\mathbf{W}}}\right]_n = \mathcal{I}_{\mathbf{W}}(\mathbf{y}_n) \quad \mathrm{for} \ n=1, \dots , N  .
\end{equation*}\end{linenomath}

\paragraph{Regularization} \ \\
Adding a quadratic norm  penalty term to the minimization functional in \eqref{eq:damas_nnlsw} yields a regularized version of DAMAS-NNLS
\begin{linenomath}\begin{equation}
\min \limits_{\mathbf{q} \geq 0} \norm{\mathbf{H}_{\mathbf{W}}\mathbf{q}-\mathbf{b}_{\mathbf{W}}}_2^2 + \alpha \norm{\mathbf{q}}_2^2 \ , \label{eq:damas_nnlsw_reg} 
\end{equation}\end{linenomath}
where $\alpha>0$ is the regularization parameter. The regularized approach \eqref{eq:damas_nnlsw_reg} ensures unique solutions and stable source power reconstructions.
%%%%%%%%%%%%%%%%%%%%%%%%%%%%%%%%%%%%%%%%%%%%%%%%%%%%%%%%%%%%%
%%%%%%%%%%%%% Theoretical Analysis %%%%%%%%%%%%%%%%%%%%%%%%%%
%%%%%%%%%%%%%%%%%%%%%%%%%%%%%%%%%%%%%%%%%%%%%%%%%%%%%%%%%%%%%

\section{Variance optimal beamforming weights} \label{sec:varopt}
In this section we will study the variance of the beamforming functionals introduced in Section \ref{sec:arrayapplication}. As a main result we will show that the beamforming functional based on the Mahalanobis distance (iv-f) minimizes the variance among all beamformers. Furthermore the differences between iv-f beamforming and \textit{Capon's Method} \cite{Capon1969} (a.k.a. \textit{Minimum Variance Method}) are discussed.  \ \\ \\
As presented in Section \ref{sec:mathmod} and \ref{sec:arrayapplication}, on any Hilbert space $\mathcal{H} = \left( \mathbb{C}^{M^2}, \langle \cdot , \cdot \rangle_{\mathbf{W}} \right) $ beamforming estimators are characterized by
\begin{linenomath}\begin{equation*}
\begin{aligned}
\mathcal{I}_{\mathbf{W}}(\mathbf{y}) 
=  \frac{ \left \langle \vect{\cobs}, \vect{\mathbf{G}(\mathbf{y}) } \right \rangle_{\mathbf{W}} } {\left \langle \vect{\mathbf{G}(\mathbf{y})} , \vect{\mathbf{G}(\mathbf{y})}  \right \rangle_{\mathbf{W}}}   \ .
\end{aligned}
\end{equation*}\end{linenomath}
Since the measured CSM is modelled as a random quantity, the estimator $\mathcal{I}_{\mathbf{W}}(\mathbf{y})$ is also random. Note that 
\begin{linenomath}\begin{equation*}
\mathbb{E}\left( \mathcal{I}_{\mathbf{W}}(\mathbf{y}) \right) = \frac{ \left \langle \vect{\mathbf{C}^{\mathrm{ac}}+ \mathbf{D}}, \vect{\mathbf{G}(\mathbf{y}) } \right \rangle_{\mathbf{W}} } {\left \langle \vect{\mathbf{G}(\mathbf{y})} , \vect{\mathbf{G}(\mathbf{y})}  \right \rangle_{\mathbf{W}}} 
\end{equation*}\end{linenomath}
is the beamforming result for ideal noise free data i.e. $\mathbf{Z}=0$. Hence it is desirable that the result for noisy data $\cobs$ does not deviate too much from $\mathbb{E}\left( \mathcal{I}_{\mathbf{W}}(\mathbf{y}) \right)$. More precisely the variance of the estimator
\begin{linenomath}\begin{equation*}
V_{\mathbf{W}}(\mathbf{y}) :=\Var{\mathcal{I}_{\mathbf{W}}(\mathbf{y})} = \mathbb{E} \left[ \abs{ \mathcal{I}_{\mathbf{W}}(\mathbf{y}) - \mathbb{E}\left(  \mathcal{I}_{\mathbf{W}}(\mathbf{y}) \right)}^2 \right] 
\end{equation*}\end{linenomath}
should be small. The next theorem states a variance optimality property of the beamforming functional based on the Mahalanobis distance. 
%\todo[inline]{TH: Ein Lemma ist ein Hilfssatz, der für sich genommen nicht von Interesse ist. Das ist hier sicher nicht der Fall. Wenn wir bescheiden sein wollen, könnten wir es Proposition nennen, aber als theoretisches Hauptresultat dieses Papers bin ich für Theorem.}
\begin{thm}[Variance optimal beamforming functional]\label{thm:varopt} \ \\
Assume that the covariance matrix of the correlation data $\mathbf{\Sigma} = \Cov{\vect{\cobs}}$ is regular, then for any Hermitian, positive definite matrix $\mathbf{W} \in \mathbb{C}^{M^2}$, the variance of the corresponding beamformer $\mathcal{I}_{\mathbf{W}}(\mathbf{y})$ is bounded from below by the variance of the iv-f beamformer i.e.
\begin{linenomath}\begin{equation*}
 \Var{\mathcal{I}_{\mathbf{\Sigma}}(\mathbf{y})} \leq \Var{\mathcal{I}_{\mathbf{W}}(\mathbf{y})} \ .
\end{equation*}\end{linenomath}
\end{thm}
\begin{proof}
To improve readability we omit the dependency on the focus point $\mathbf{y}$. Consider the elementary rule for linear transformations of complex-valued random variables  \cite[Corollary 1.1]{Andersen1995} which states that for a random variable $\mathbf{x}$ mapping to $\mathbb{C}^d$ and a matrix $\mathbf{K} \in \mathbb{C}^{n \times d} $ we have
\begin{linenomath}\begin{equation*} 
\Cov{\mathbf{K} \mathbf{x}}=  \mathbf{K}\Cov{\mathbf{x}} \mathbf{K}^*   \ .
\end{equation*}\end{linenomath}
For the choice
\begin{linenomath}\begin{equation*}
\begin{aligned}
\mathbb{C}^{M^2 \times 1} \ni \mathbf{x} &=  \vect{ \cobs} \\
\mathbb{C}^{1 \times M^2} \ni \mathbf{K} &= \frac{1}{\left \langle \vG ,  \mathbf{W}^{-1}\vG  \right \rangle_2}  \vG^*  \mathbf{W}^{-1}
\end{aligned}
\end{equation*}\end{linenomath} 
we obtain
\begin{linenomath}\begin{equation} \label{eq:varw_def}
\begin{aligned}
V_{\mathbf{W}} &= \Cov{\mathbf{K} \mathbf{x}} = \frac{ \left \langle \mathbf{W}^{-1} \vG , \mathbf{\Sigma} \mathbf{W}^{-1} \vG  \right \rangle_{2} } { \left \langle \vG ,  \mathbf{W}^{-1} \vG  \right \rangle_{2}^2 }   \\
&= \frac{ \norm{\mathbf{W}^{-1} \vect{\mathbf{G}}}_{\mathbf{\Sigma}^{-1}}^2 } { \left \langle \vG ,  \mathbf{W}^{-1} \vG  \right \rangle_{2}^2 }
\ .
\end{aligned}
\end{equation}\end{linenomath} 
Using the Cauchy-Schwarz inequality on the denominator in \eqref{eq:varw_def} yields 
\begin{linenomath}\begin{equation} \label{eq:varw_lowerbound}
\begin{aligned}
V_{\mathbf{W}} &= \frac{ \norm{\mathbf{W}^{-1} \vect{\mathbf{G}}}_{\mathbf{\Sigma}^{-1}}^2} { \left \langle \mathbf{\Sigma}^{\nicefrac{-1}{2}} \vG ,  \mathbf{\Sigma}^{\nicefrac{1}{2}}\mathbf{W}^{-1} \vG  \right \rangle_{2}^2 } \\
&\geq \frac{ \norm{\mathbf{W}^{-1} \vect{\mathbf{G}}}_{\mathbf{\Sigma}^{-1}}^2 } { \left \langle \mathbf{\Sigma}^{\nicefrac{-1}{2}} \vG ,  \mathbf{\Sigma}^{\nicefrac{-1}{2}} \vG  \right \rangle_{2}  \left \langle \mathbf{\Sigma}^{\nicefrac{1}{2}}\mathbf{W}^{-1} \vG  ,  \mathbf{\Sigma}^{\nicefrac{1}{2}}\mathbf{W}^{-1} \vG  \right \rangle_{2}   } \\
&= \frac{ \norm{\mathbf{W}^{-1} \vect{\mathbf{G}}}_{\mathbf{\Sigma}^{-1}}^2 } { \norm{ \vect{\mathbf{G}}}_{\mathbf{\Sigma}}^2 \norm{\mathbf{W}^{-1} \vect{\mathbf{G}}}_{\mathbf{\Sigma}^{-1}}^2} \\
&= \frac{ 1 } { \norm{ \vect{\mathbf{G}}}_{\mathbf{\Sigma}}^2 } \ .
\end{aligned}
\end{equation}\end{linenomath}
Inserting $\mathbf{W} = \mathbf{\Sigma}$ in \eqref{eq:varw_def} yields
\begin{linenomath}\begin{equation*}
V_{\mathbf{\Sigma}} = \frac{ 1 } { \norm{ \vect{\mathbf{G}}}_{\mathbf{\Sigma}}^2 }
\end{equation*}\end{linenomath}
i.e. the lower bound in \eqref{eq:varw_lowerbound} is attained for $\mathbf{W} = \mathbf{\Sigma}$.
\end{proof}
Theorem \ref{thm:varopt} shows that the beamformer for the weighting choice \eqref{eq:wc_iv-f} minimizes the variance among all possible choices. Due to this minimizing property one may ask if there is a relation to Capon's Method which is often called \textit{Minimum Variance Method} in the literature. We want to emphasize that the variance minimizing property of Capon's method is not the same as the minimizing property of Theorem \ref{thm:varopt}. An alternative definition of Capon's method is
\begin{linenomath}\begin{equation*}
\mathcal{I}^{\text{Cap}} = \mathbf{w}_{\text{Cap}}^* \cobs \mathbf{w}_{\text{Cap}} \ .
\end{equation*}\end{linenomath}
The weight vector $\mathbf{w}_{\text{Cap}}$ is defined as the solution of the minimization problem 
\begin{linenomath}\begin{equation}
\min \limits_{\mathbf{w} \in \mathbb{C}^{M}} \mathbf{w}^* \cobs \mathbf{w}  \quad \text{subject to} \quad \real{\mathbf{g}^* \mathbf{w} } = 1 \ . \label{eq:capon_wmin}
\end{equation}\end{linenomath}
Note that for $\mathbf{p} = \left(p(\mathbf{x}_1), \dots , p(\mathbf{x}_M) \right)^{\top}$ it holds that 
\begin{linenomath}\begin{equation*}
\mathbb{E}\left[ \abs{\mathbf{w}^* \mathbf{p}}^2  \right] = \mathbf{w}^*\mathbb{E}\left[ \mathbf{p}\mathbf{p}^*  \right] \mathbf{w} \approx \mathbf{w}^* \cobs \mathbf{w} \ .
\end{equation*}\end{linenomath}
The mean squared value $\mathbb{E}\left[ \abs{\mathbf{w}^* \mathbf{p}}^2  \right]$ equals the variance $\Var{\mathbf{w}^* \mathbf{p}}$ if $\mathbf{p}$ has zero mean. On the one hand that explains the commonly used term \textit{Minimum Variance Method} (see also \cite{Johnson1993}) and on the other hand it shows that the variance that is minimized for Capon's method $\Var{\mathbf{w}^* \mathbf{p}}$ differs from the variance $\Var{\mathcal{I}_{\mathbf{W}}}$, considered in Theorem \ref{thm:varopt}. In particular Capon's method and the iv-f beamforming method do not yield the same results in general. However, under additional assumptions on the data, both methods yield similar results, as we will demonstrate in the remainder of this section. \ \\ \\
For $\cobs$ being invertible, the solution to \eqref{eq:capon_wmin} is explicitly given by
\begin{linenomath}\begin{equation*}
\mathbf{w}_{\text{Cap}} = \frac{\left(\cobs\right)^{-1} \mathbf{g}}{\mathbf{g}^*\left(\cobs\right)^{-1} \mathbf{g} } \ ,
\end{equation*}\end{linenomath}
which yields the definition of Capon's method given in Eq. \eqref{eq:caponbf_def}. For the remaining analysis within this section we assume that $\cobs$ is positive definite (and thus invertible) along with the subsequent additional assumptions. 
\begin{ass}[Pressure signals] \label{as:proper} \ \\The pressure signal vector $\mathbf{p} = \left(p(\mathbf{x}_1), \dots , p(\mathbf{x}_M) \right)^{\top}$ has the following properties 
%\vspace{-1\baselineskip}
\begin{enumerate}[(i)]
\item $\mathbf{p}$ is a complex $M$-dimensional Gaussian random variable \label{it:fst_p},
\item $\mathbb{E}\left[  \mathbf{p} \right] = 0 $ \label{it:scnd_p},
\item $\mathbb{E}\left[  \mathbf{p} \mathbf{p}^{\top} \right] = \mathbf{0} $ \ .  \label{it:thrd_p} 
\end{enumerate}
\end{ass}
Note that property (\ref{it:scnd_p}) follows directly under Assumption \ref{as:cor}.\ref{it:zermean_parts}. Whereas properties (\ref{it:fst_p}) and (\ref{it:scnd_p}) are quite common and intuitive, property (\ref{it:thrd_p}) seems less intuitive and might also be too restrictive. Complex random variables that fulfill property (\ref{it:thrd_p}) are usually called \textit{proper} \cite[Def. 2.1, p.35]{Schreier2014}. Due to Assumption \ref{as:proper} (\ref{it:fst_p})-(\ref{it:scnd_p}), we can apply the following explicit formula for the covariances of cross correlations 
\begin{linenomath}\begin{equation} \label{eq:covcov_formula}
\begin{aligned}
& \ \ \ \ \Cov{ p(\mathbf{x}_m) p(\mathbf{x}_l)^*,p(\mathbf{x}_{m'})p(\mathbf{x}_{l'})^*    } \\
&= \mathbb{E}\left(p(\mathbf{x}_m) p(\mathbf{x}_{m'})^* \right)  \mathbb{E}\left(p(\mathbf{x}_l) p(\mathbf{x}_{l'})^* \right)^* \\
&+ \mathbb{E}\left(p(\mathbf{x}_m)p(\mathbf{x}_{l'}) \right)  \mathbb{E}\left(p(\mathbf{x}_l)p(\mathbf{x}_{m'}) \right)^*  \ .
\end{aligned} 
\end{equation}\end{linenomath}
This result follows from Isserlis' theorem \cite{Isserlis1916, Isserlis1918}, see \cite{Gizon2004} and \cite{Fournier2014} for a more detailed discussion. Furthermore, due to Assumption \ref{as:proper} (\ref{it:thrd_p}), the second summand in \eqref{eq:covcov_formula} vanishes and hence
\begin{linenomath}\begin{equation} \label{eq:covcov_formula_proper}
\begin{aligned}
\Cov{ p(\mathbf{x}_m) p(\mathbf{x}_l)^*,p(\mathbf{x}_{m'})p(\mathbf{x}_{l'})^*    } 
&= \mathbb{E}\left(p(\mathbf{x}_m) p(\mathbf{x}_{m'})^* \right)  \mathbb{E}\left(p(\mathbf{x}_l) p(\mathbf{x}_{l'})^* \right)^* \\
& \approx \cobs_{mm'} \cobs_{l'l} \ .
\end{aligned}
\end{equation}\end{linenomath}
Relation \eqref{eq:covcov_formula_proper} reveals that we can estimate the covariance matrix $\Cov{\vect{\mathbf{C}}^{\text{obs}}}$ by
\begin{linenomath}\begin{equation*}
\mathbf{\Sigma}^{\text{est}} = \left( \cobs \right)^{\top} \otimes \cobs \ ,
\end{equation*}\end{linenomath}
which is exactly the weighting matrix that represents Capon's method (see Eq. \eqref{eq:wc_capon}). Hence, for proper Gaussian pressure signals with zero mean, Capon's method and beamforming based on the Mahalanobis distance yield the same result if the covariance matrix $\mathbf{\Sigma}$ is estimated according to \eqref{eq:covcov_formula_proper}. Again we want to emphasize that in general both methods are not equivalent and do not yield  the same results. However, if the data satisfy all three items of Assumption \ref{as:proper} with sufficient accuracy, Capon's method is a good approximation of the iv-f beamformer. In Section \ref{sec:results-experimental} we will discuss some statistical measures regarding the plausibility of the properties (\ref{it:fst_p})-(\ref{it:thrd_p}) for an experimental dataset.

%%%%%%%%%%%%%%%%%%%%%%%%%%%%%%%%%%%%%%%%%%%%%%%%%%%%%%%%%%%%%
%%%%%%%%%%%%% Computational Aspects %%%%%%%%%%%%%%%%%%%%%%%%%
%%%%%%%%%%%%%%%%%%%%%%%%%%%%%%%%%%%%%%%%%%%%%%%%%%%%%%%%%%%%%

\section{Computational aspects}\label{sec:computational}
In this section we will discuss some relevant computational aspects for the implementation of the presented methods.
\subsection{Removal of sensor pairs}
For any subset of indices $R \subset \lbrace 1, \dots , M \rbrace^2 $ (see also \cite{Sijtsma2004}) one may remove all components from the vectorizations $\vect{\cobs}$ and $\vect{\mathbf{G}}$ that correspond to an index pair $(m,l) \in R$. Similarly one removes all corresponding rows and columns of the weighting matrix $\mathbf{W}$. All presented methods can now be stated with respect to those reduced quantities. If no further knowledge on the values of the boundary layer noise matrix $\mathbf{D}$ is imposed, a common choice for the removal indices is $R = \lbrace (m,l): \ m = l \rbrace$ i.e. the diagonal of the measured CSM is not taken into account. This is also known as \textit{diagonal removal}. Note that the result of Theorem \ref{thm:varopt} remains valid for any choice of $R$.
\subsection{Broad band source maps}
In order to reduce noise effects in the imaging results, beamforming functionals are often averaged over frequency bands. More precisely for a center frequency $\omega_0$, a surrounding frequency band $B(\omega_0)$ and an imaging functional $\mathcal{I}$, the averaged source map at a focus point $\mathbf{y}$ yields the value
\begin{linenomath}\begin{equation*}
\int_{B(\omega_0)} \mathcal{I}(\mathbf{y}, \omega) d \omega \ .
\end{equation*}\end{linenomath}
A popular choice for $B$ is the third octave band according to ISO 266:1997 \cite{enISO266}
\begin{linenomath}\begin{equation*}
B(\omega_0) =  \left[2^{\nicefrac{-1}{6} }  \omega_0, \ 2^{\nicefrac{1}{6} }  \omega_0 \right] \ .
\end{equation*}\end{linenomath}

\subsection{Covariance estimation}
In order to apply weighting matrices $\mathbf{W}$ that depend on data (co)variances, an appropriate estimator of the noise covariance matrix $\mathbf{\Sigma}$ is needed. 
\paragraph{Estimation for zero mean Gaussian signals:} \ \\
Under Assumption \ref{as:proper} (\ref{it:fst_p})-(\ref{it:scnd_p}), we can employ the covariance Formula \eqref{eq:covcov_formula}. Then expressions of the type $\mathbb{E}\left(p(\mathbf{x}_m) p(\mathbf{x}_{l})^* \right)$ can be estimated by the corresponding entry of the CSM $\cobs$. Expressions of the type $\mathbb{E}\left(p(\mathbf{x}_m) p(\mathbf{x}_{l}) \right)$ can be estimated by the corresponding entry of the measured pseudo cross spectral matrix (PCSM) $\mathbf{C}^{\text{ps}}$ 
\begin{linenomath}\begin{equation} \label{eq:pcsm_def}
\mathbf{C}^{\text{ps}}_{ml} = \mathbb{M} \left\lbrace p_j(\mathbf{x}_m)p_j(\mathbf{x}_l) \right\rbrace = \frac{1}{J}\sum \limits_{j=1}^J  p_j(\mathbf{x}_m)p_j(\mathbf{x}_l) \ .
\end{equation}\end{linenomath}
Note that Formula \eqref{eq:covcov_formula} ensures that the estimated covariance matrix is Hermitian but it does not ensure positive semi-definiteness or regularity. 
\paragraph{Estimation by sample covariances:} \ \\
Without any additional assumptions the estimation may be done by sample covariances of the CSM entries. The number of block samples is denoted by $J$. For a block sample index $j$ denote by $\cobs_j = \mathbf{p}_j\mathbf{p}_j^*$ the $j-$th CSM sample. $\mathbf{\Sigma}$ may then be estimated by
\begin{linenomath}\begin{equation} \label{eq:covcov_sample}
\mathbf{\Sigma} \approx \mathbf{\Sigma}^{\text{samp}} = \frac{1}{J} \sum \limits_{j=1}^J \left[ \vect{\cobs_j} - \vect{\cobs}  \right] \left[ \vect{\cobs_j} - \vect{\cobs}  \right]^* \ .
\end{equation}\end{linenomath}
Since each summand is a rank-one matrix and the rank of the sum can be bounded from above by the sum over the ranks we obtain
\begin{linenomath}\begin{equation}
\text{rank}\left( \mathbf{\Sigma}^{\text{samp}} \right) \leq J \ . \label{eq:rankineq}
\end{equation}\end{linenomath}
One has to ensure that the estimated covariance matrix is regular. By \eqref{eq:rankineq} this cannot hold true for a small number of block samples $J < M^2$ (where $M$ denotes the number of microphones). Since $M^2 \sim 10^4$ for many aeroacoustic measurement setups, the number of block samples will not be sufficiently large in many cases. 
\paragraph{Definiteness and regularity:} \ \\
For a covariance estimator $\mathbf{\Sigma^{\mathrm{est}}}$ (derived by \eqref{eq:covcov_formula} or \eqref{eq:covcov_sample}) that is not positive semi-definite or invertible one may proceed as follows:
Choose a regularization parameter $\alpha>0$ and solve 
%the semi-definite optimization problem
\begin{linenomath}\begin{equation} \label{eq:closest_cov}
\begin{aligned}
\mathbf{\Sigma}^{\star} = \alpha \mathbf{I} \ + \ 
&\argmin \limits_{\mathbf{\tilde{\Sigma}^{\star}}} \norm{ \mathbf{\tilde{\Sigma}^{\star}} + \alpha \mathbf{I} - \mathbf{\Sigma}^{\mathrm{est}}}^2_F  \\ &\text{subject to} \ \ \mathbf{\tilde{\Sigma}^{\star}} \ \text{is positive semi-definite}.
\end{aligned}
\end{equation}\end{linenomath} 
Then $\mathbf{\Sigma}^{\star}$ is the best approximation of $\mathbf{\Sigma}^{\mathrm{est}}$ among all positive semi-definite matrices whose inverse is bounded by $\frac{1}{\alpha}$.
The minimization problem in \eqref{eq:closest_cov} is uniquely solvable \cite[Theorem 8.8]{Higham2008} and may be solved numerically by the methods described in \cite[Chapter 8]{Higham2008} or by an optimization toolbox such as CVX \cite{Grant2008, Grant2014} or SDPT$^3$ \cite{Toh1999}.

\subsection{Computational effort}
The number of microphones is usually of the order $M \sim 10^2$. Considering the imaging functional evaluations for a diagonal weighting choice e.g. iv-d \eqref{eq:wc_iv-d} or conventional \eqref{eq:wc_conventional} we conclude that the computational complexity of one evaluation (i.e. beamforming to one focus point) is $\mathcal{O}\left( M^2 \right)$. For non-diagonal weightings e.g. iv-f \eqref{eq:wc_iv-f} the evaluation step requires the solution of a linear system with a system matrix of size $M^2 \times M^2$. If a direct solver is applied the computational complexity is $\mathcal{O}\left( M^{2 \gamma} \right)$ for some $\gamma \in (2, 3]$ depending on the implementation of the linear system solver. This shows a significantly increased effort for non-diagonal weightings if the system matrix is dense. However, to reduce the computational effort in that case one may approximate the system matrix by a low-rank perturbation of a diagonal matrix 
\begin{linenomath}\begin{equation*}
    \mathbf{W} \approx \mathbf{D} + \mathbf{L}^*\mathbf{L} \ ,
\end{equation*}\end{linenomath}
where $\mathbf{D}$ is diagonal and $\mathbf{L} \in \mathbb{C}^{L \times M^2}$ with $L \ll M^2$. The Sherman-Morrison-Woodbury formula then implies
\begin{linenomath}\begin{equation*}
\left( \mathbf{D} + \mathbf{L}^*\mathbf{L} \right)^{-1} = \mathbf{D}^{-1} - \mathbf{D}^{-1} \mathbf{L}^* \left(\mathbf{I} + \mathbf{L} \mathbf{D}^{-1} \mathbf{L}^*  \right)^{-1} \mathbf{L} \mathbf{D}^{-1} \ .
\end{equation*}\end{linenomath}
That allows a more efficient solution of the linear system, since the term inside the round brackets is a $L \times L$ matrix. \ \\ 
Note that the evaluation of the imaging functional can easily be parallelized since the evaluations for different map points and frequencies are completely decoupled.

\subsection{Choice of the regularization parameter}
For the regularized version of DAMAS-NNLS \eqref{eq:damas_nnlsw_reg}, a parameter choice rule for the regularization parameter $\alpha$ is needed. One of the most well-known rules is provided by  \textit{Morozov's discrepancy principle} \cite{Morozov1968}. Assume that a measure of the data noise level $\delta$ is a priori known and let $\mathbf{q}_{\alpha}$ denote the solution of the regularized DAMAS-NNLS problem \eqref{eq:damas_nnlsw_reg}. The discrepancy principle states to choose $\alpha$ by
\begin{linenomath}\begin{equation}
    \alpha_{\delta} =  \sup \lbrace \alpha>0: \ \norm{\mathbf{H}_{\mathbf{W}}\mathbf{q}_{\alpha}-\mathbf{b}_{\mathbf{W}}}_2 \leq \tau \delta  \rbrace  \label{eq:discrepancy}
\end{equation}\end{linenomath}
with some constant $\tau \geq 1$. Since the discrepancy principle tends to lead to oversmoothing for random noise and large $N$, the constant $\tau$ should not be chosen much greater than $1$ as long as stochastic noise dominates systematic errors. The (stochastic) noise level of a beamforming map is measured by the root mean squared deviation
\begin{linenomath}\begin{equation*}
    \delta^{\mathrm{rms}}_{\mathbf{W}} = \sqrt{\sum \limits_{n=1}^N \Var{\mathcal{I}_{\mathbf{W}}(\mathbf{y}_n)} }  \ . 
\end{equation*}\end{linenomath}
With $ \tilde{\mathbf{\Sigma}} = \Cov{\vect{\mathbf{p}\mathbf{p}^*}} = J \mathbf{\Sigma}$ , the root mean square deviation is equivalently represented by
\begin{linenomath}\begin{equation*}
    \delta^{\mathrm{rms}}_{\mathbf{W}} = J^{-\nicefrac{1}{2}} \sqrt{ \mathlarger{\mathlarger{\sum}}_{n=1}^N \frac{ \left \langle \mathbf{W}^{-1} \vGn , \tilde{\mathbf{\Sigma}} \mathbf{W}^{-1} \vGn  \right \rangle_{2} } { \left \langle \vGn ,  \mathbf{W}^{-1} \vGn  \right \rangle_{2}^2 } }  \ ,
\end{equation*}\end{linenomath}
where the second factor is independent of the number of block samples $J$. Moreover, we choose $\tau =1.5$, which effectively allows for systematic errors of size smaller than $50\%$ of the size of the stochastic noise. 

%%%%%%%%%%%%%%%%%%%%%%%%%%%%%%%%%%%%%%%%%%%%%%%%%%%%%%%%%%%%%
%%%%%%%%%%%%%%%%%% Results- synthetic %%%%%%%%%%%%%%%%%%%%%%%
%%%%%%%%%%%%%%%%%%%%%%%%%%%%%%%%%%%%%%%%%%%%%%%%%%%%%%%%%%%%%
 
\section{Results on Synthetic Data} \label{sec:results-synthetic}
In this section we test the iv-d and iv-f weighting on a simple synthetic dataset with a single monopole source. The results are compared to those obtained by a standard weighting with respect to the potential improvements regarding resolution and the reduction of noise effects. The data is generated using an array of 64 microphones that has been previously used for synthetic benchmark exercises \cite{Sarradj2017} and performance analysis of microphone array methods \cite{Herold2017}. The array has an aperture of $1.5$m and is located on a x-y-plane centered at the origin. The focus plane is located at $z=0.75\,\mathrm{m}$ and $x,y \in [-0.5\, \mathrm{m}, 0.5\,\mathrm{m}]$ and discretized by an equidistant grid with a resolution of $\Delta x = \Delta y = 0.025\,\mathrm{m}$ . 
\ \\ \ \\
We consider a single monopole source located at the center of the focus region $(0,0,0.75)^{\top}$. The ensemble consists of $J=1000$ independent pressure samples of the form
\begin{linenomath}\begin{equation}
\mathbf{p}^{(j)} =  \eta^{(j)} p_0 \mathbf{g}(\mathbf{0}) + \rho \boldsymbol{\epsilon}^{(j)} , \label{eq:synth-samples}
\end{equation}\end{linenomath}
where $\eta^{(j)}$ is drawn from a one-dimensional standard complex normal distribution and $\boldsymbol{\epsilon}^{(j)}$ from a M-dimensional standard complex normal distribution. By this procedure, we make sure that the pressure samples are complex Gaussian with zero mean.
For the additive noise power $\rho$ we consider three levels such that
\begin{linenomath}\begin{equation*}
\Delta_{\mathrm{noise}} := 20 \log_{10}\left(\frac{p_0}{\rho} \right) \in \lbrace 20, 10,0 \rbrace \ .
\end{equation*}\end{linenomath}
For each dataset beamforming and DAMAS-NNLS solutions are computed. The regularization parameter is chosen by the discrepancy principle with $\tau = 1$ and $\delta = \delta_{\mathbf{W}}^{\mathrm{rms}}$. \ \\ \ \\
To compare the results of different weightings, three measures are used, one for the resolution, one for the SNR and one that combines effects on resolution and lowered noise. Those measures were proposed in \cite{Lehmann2020pre} and are defined in the following.
%%%%%%%%%%%%%%%% Resolution %%%%%%%%%%%%%%%%%%%%%%%%%%%%
\paragraph{Resolution measure}
Let $\mathbf{y}_s$ denote the position of the global sourcemap maximum and $\mathcal{L}_{1\mathrm{dB}}$ the level set of the value $1$dB below the maximum.
The Resolution measure is then given by
\begin{linenomath}\begin{equation}
\mathrm{Resolution} = \sup_{\mathbf{y} \in \mathcal{L}_{1\mathrm{dB}}} \norm{\mathbf{y}_s - \mathbf{y}} \label{eq:res_measure}
\end{equation}\end{linenomath}
i.e. the largest distance from the maximum location to the -1dB levelset.

\begin{figure}[h]
\begin{subfigure}{.32\textwidth}
  \centering
  \includegraphics[width=\textwidth]{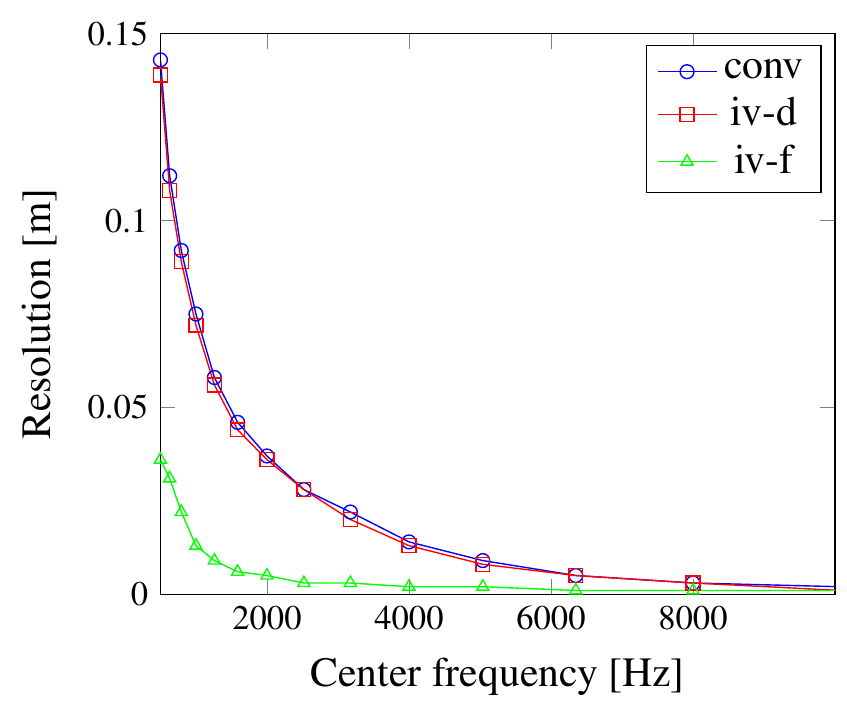}
  \caption{$\Delta_{\mathrm{noise}} = 20$dB}
\end{subfigure}
\begin{subfigure}{.32\textwidth}
  \centering
  \includegraphics[width=\textwidth]{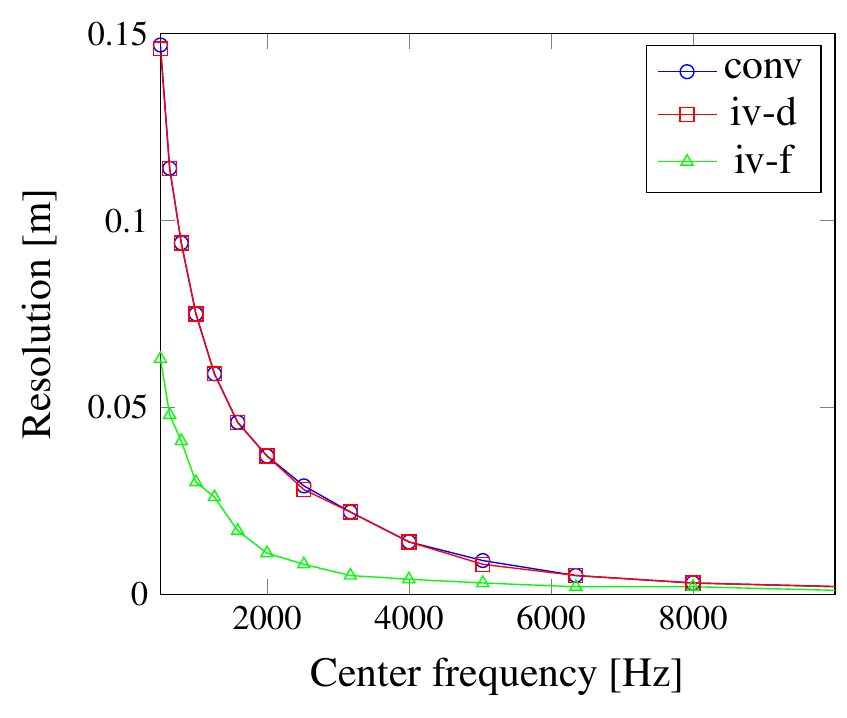}
  \caption{$\Delta_{\mathrm{noise}} = 10$dB}
\end{subfigure}
\begin{subfigure}{.32\textwidth}
  \centering
  \includegraphics[width=\textwidth]{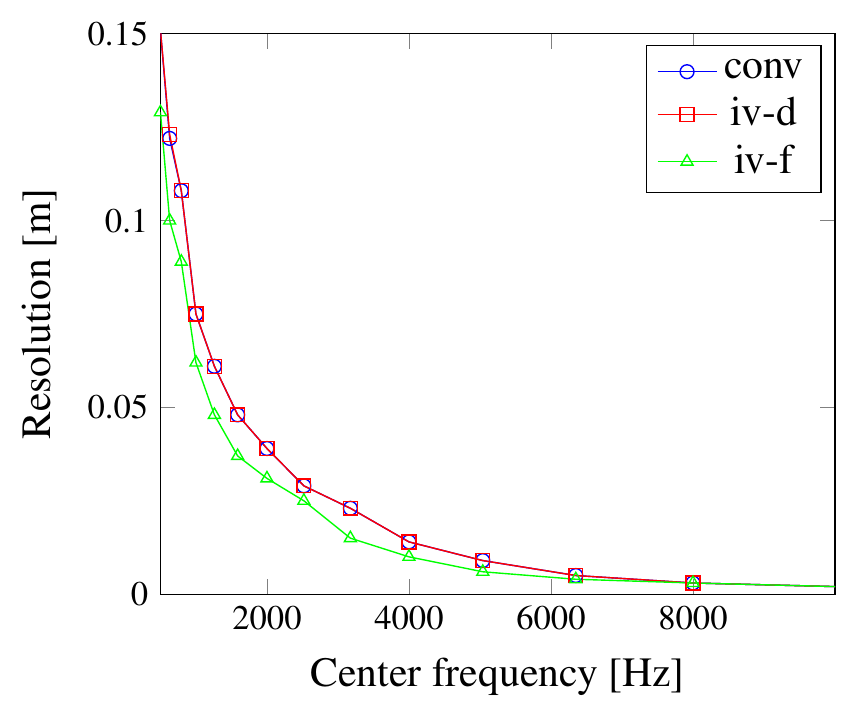}
  \caption{$\Delta_{\mathrm{noise}} = 0$dB}
\end{subfigure}%
\\
\vspace*{0.1cm}
\\
\begin{subfigure}{.32\textwidth}
  \centering
  \includegraphics[width=\textwidth]{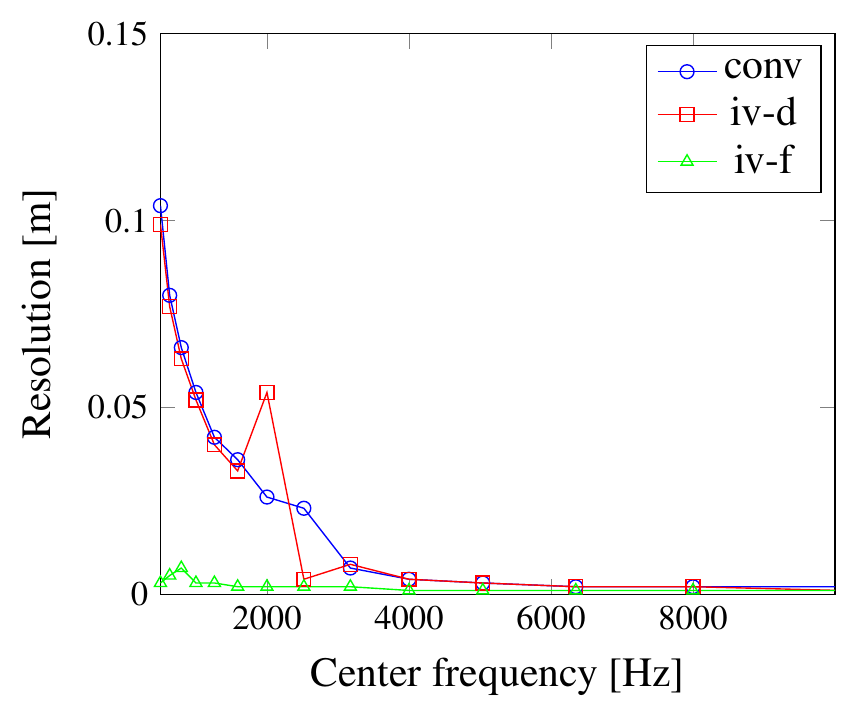}
  \caption{$\Delta_{\mathrm{noise}} = 20$dB}
\end{subfigure}
\begin{subfigure}{.32\textwidth}
  \centering
  \includegraphics[width=\textwidth]{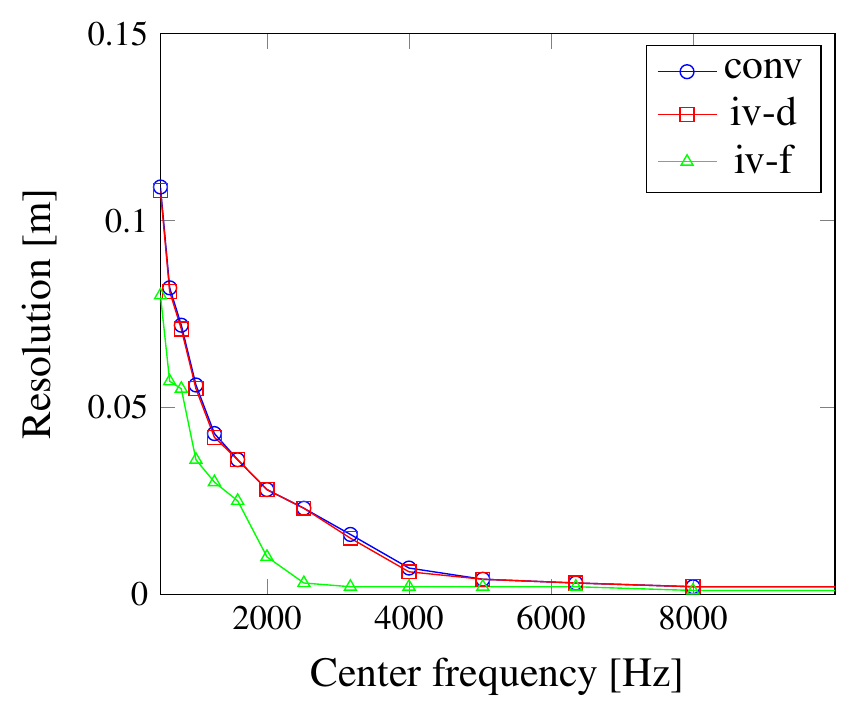}
  \caption{$\Delta_{\mathrm{noise}} = 10$dB}
\end{subfigure}
\begin{subfigure}{.32\textwidth}
  \centering
  \includegraphics[width=\textwidth]{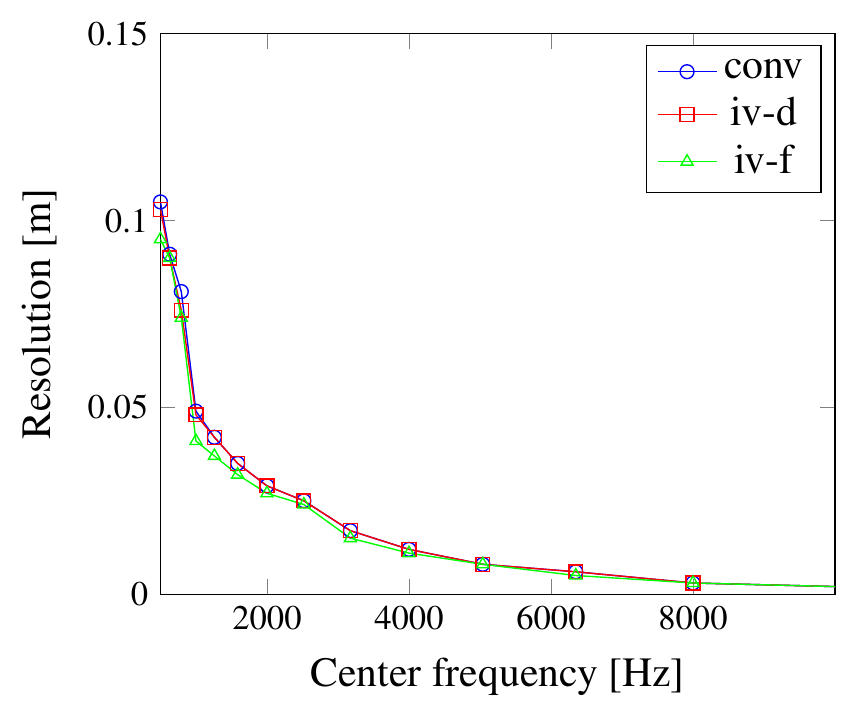}
  \caption{$\Delta_{\mathrm{noise}} = 0$dB}
\end{subfigure}%
\caption{Resolution measure for beamforming (a)-(c) and regularized DAMAS-NNLS with discrepancy principle (d)-(f).}
\label{fig:res-maps}
\end{figure}

%%%%%%%%%%%%%%%% SNR %%%%%%%%%%%%%%%%%%
\paragraph{SNR measure}
As a measure for the SNR we compute the ratio of the main lobe to the largest sidelobe. Without changing the result we can normalize each sourcemap $S$ such that the maximum source power level is always at $0$dB. For a given sound pressure level $L < 0$ dB we define
\begin{linenomath}\begin{equation*}
D(L) = \lbrace \mathbf{y} \in \mathcal{Y} : \ S(\mathbf{y}) > L \rbrace
\end{equation*}\end{linenomath} 
the set of points, where the sourcemap is greater than the level $L$. By $\mathrm{concomp}(D(L))$ we denote the number of connected components of $D(L)$. The maximum sidelobe is then given by
\begin{linenomath}
\begin{equation}
\mathrm{SNR} = \abs{\sup\lbrace L<0: \ \mathrm{concomp}(D(L))>1  \rbrace}\ . \label{eq:snr_measure}
\end{equation}
\end{linenomath}
To illustrate this metric one may imagine a sourcemap where all regions with level lower than $L$ are hidden. Then one reduces $L$ as long as the visible region is a single connected domain i.e. the main lobe. As soon as a second disjoint domain (the largest sidelobe) appears, the reduction is stopped and the current absolute value of $L$ yields the SNR measure. 
\begin{figure}[h]
\begin{subfigure}{.32\textwidth}
  \centering
  \includegraphics[width=\textwidth]{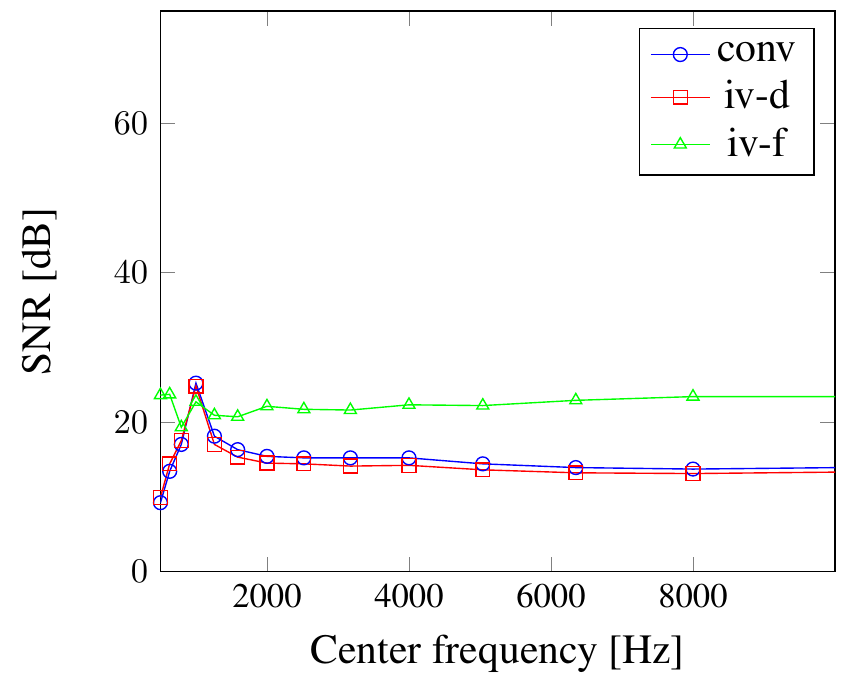}
  \caption{$\Delta_{\mathrm{noise}} = 20$dB}
\end{subfigure}
\begin{subfigure}{.32\textwidth}
  \centering
  \includegraphics[width=\textwidth]{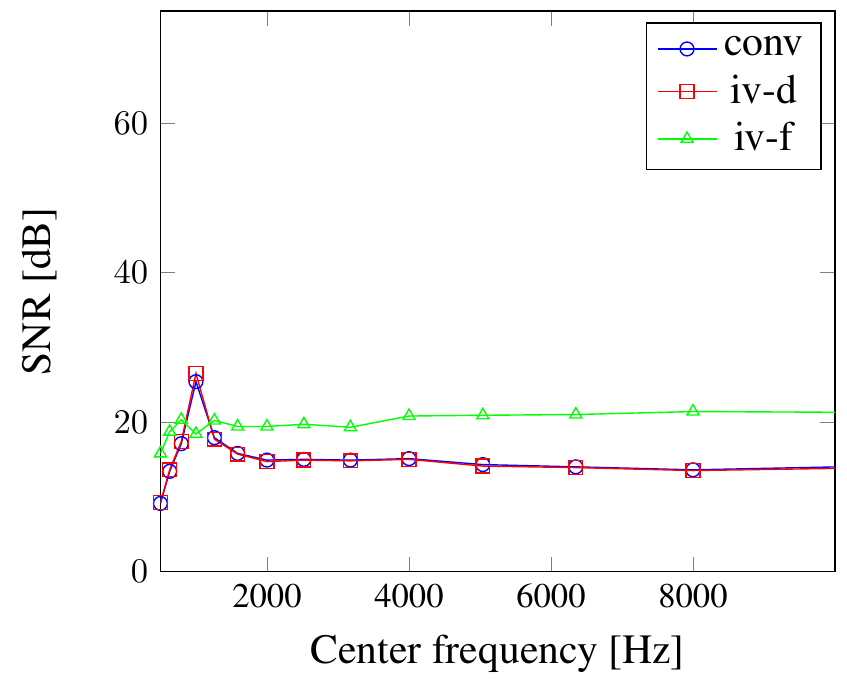}
  \caption{$\Delta_{\mathrm{noise}} = 10$dB}
\end{subfigure}
\begin{subfigure}{.32\textwidth}
  \centering
  \includegraphics[width=\textwidth]{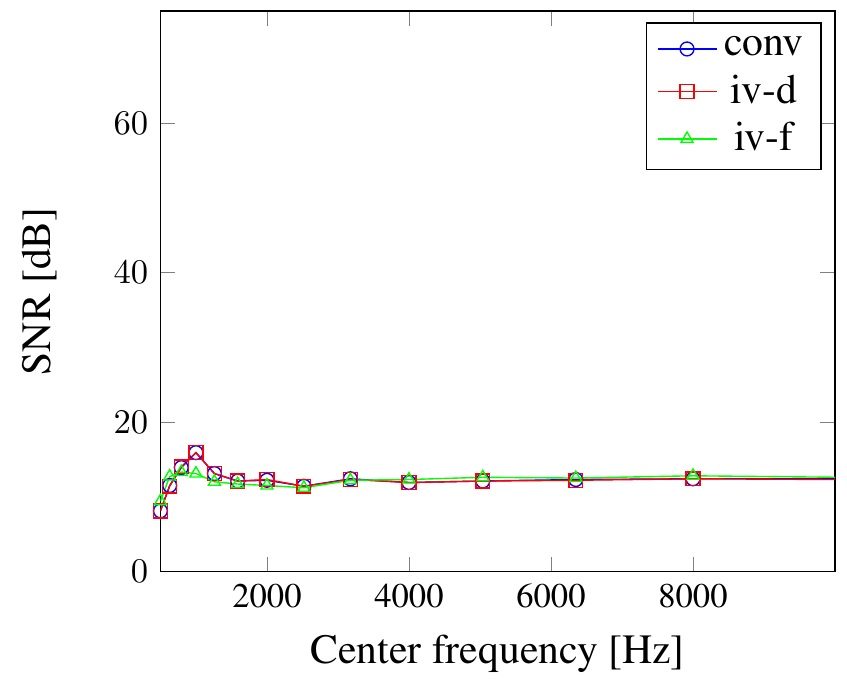}
  \caption{$\Delta_{\mathrm{noise}} = 0$dB}
\end{subfigure}%
\\
\vspace*{0.1cm}
\\
\begin{subfigure}{.32\textwidth}
  \centering
  \includegraphics[width=\textwidth]{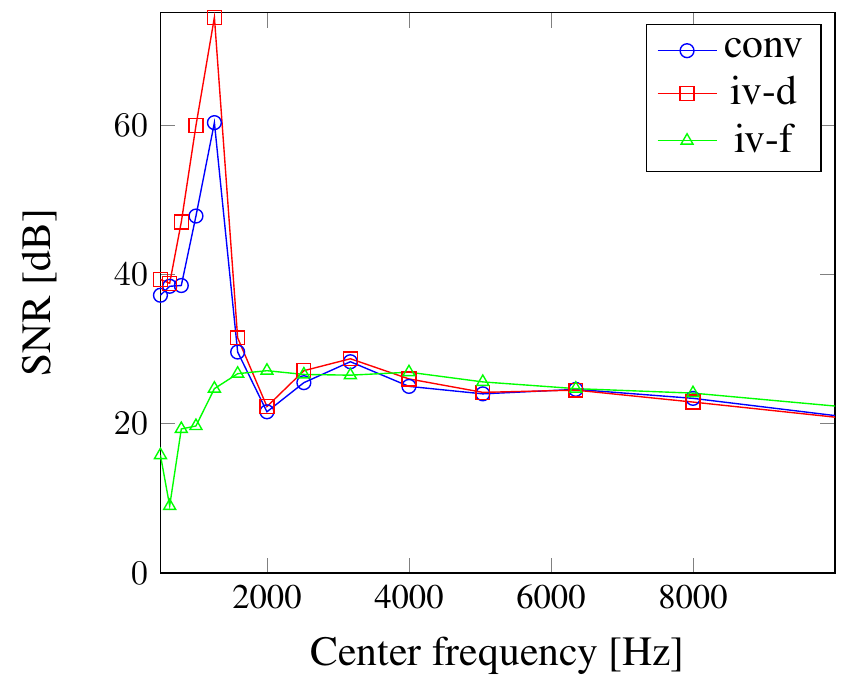}
  \caption{$\Delta_{\mathrm{noise}} = 20$dB}
\end{subfigure}
\begin{subfigure}{.32\textwidth}
  \centering
  \includegraphics[width=\textwidth]{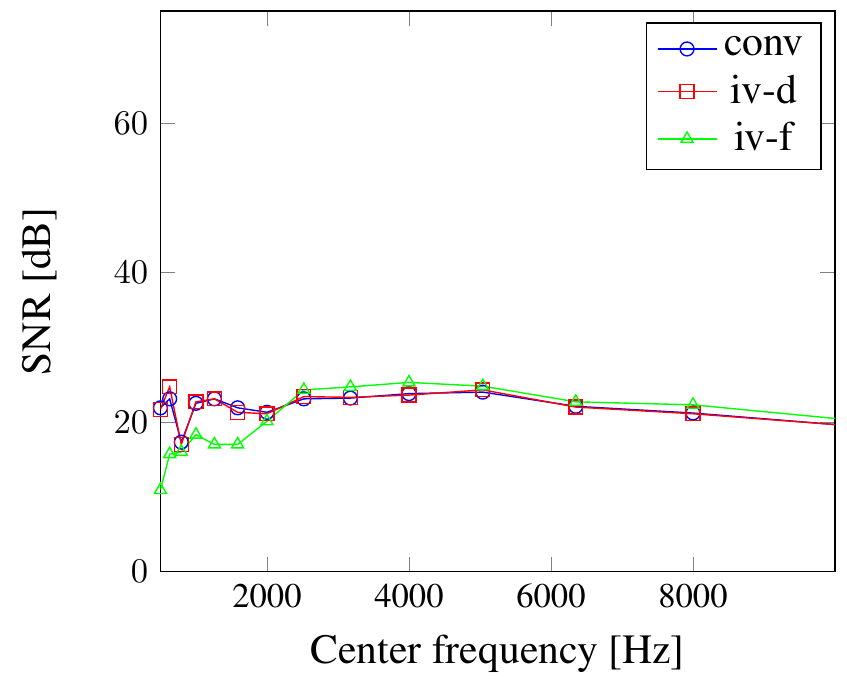}
  \caption{$\Delta_{\mathrm{noise}} = 10$dB}
\end{subfigure}
\begin{subfigure}{.32\textwidth}
  \centering
  \includegraphics[width=\textwidth]{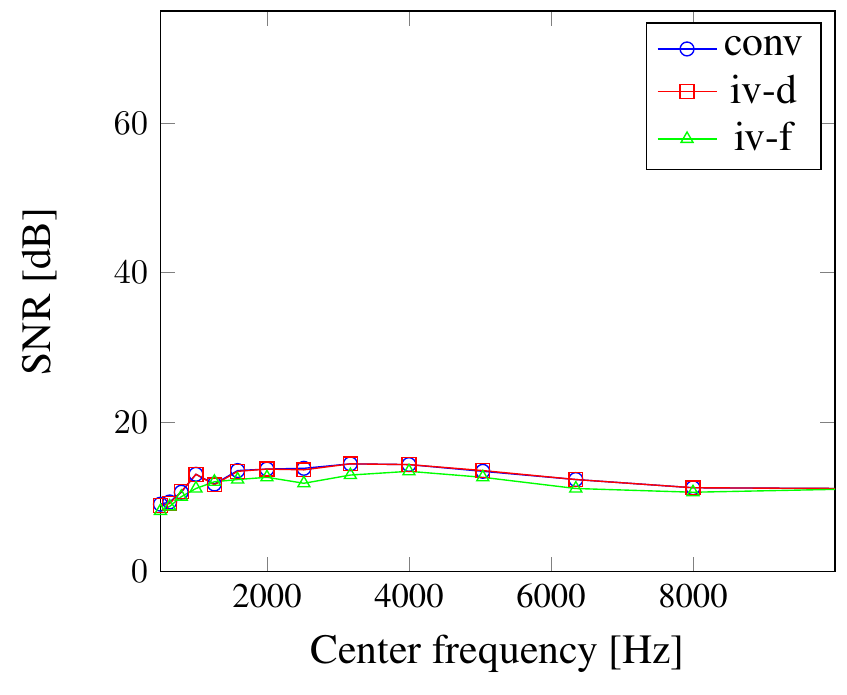}
  \caption{$\Delta_{\mathrm{noise}} = 0$dB}
\end{subfigure}%
\caption{SNR measure for beamforming (a)-(c) and regularized DAMAS-NNLS with discrepancy principle (d)-(f).}
\label{fig:snr-maps}
\end{figure}
%%%%%%%%%%%%%%%% Overall %%%%%%%%%%%%%%%%%%%%%%%%%%%%
\paragraph{Combined measure on Resolution and SNR}
The last measure is designed to provide an indicator on the overall quality of the imaging result. It is defined as the ratio of the maximum source power level to the average source power level of the entire sourcemap and will be denoted by \textit{source-to-pattern ratio} (SPR). For a sourcemap $S$ that maps each focus point $\mathbf{y}_n$ to a source power estimator $S(\mathbf{y}_n)$ this yields
\begin{linenomath}
\begin{equation}
    \mathrm{SPR} = 10\log_{10}\left( \frac{S(\mathbf{y}_s)}{\frac{1}{N} \sum_{n=1}^N  S(\mathbf{y}_n)} \right) \ . \label{eq:spr_measure}
\end{equation}
\end{linenomath}
\begin{figure}[h]
\begin{subfigure}{.32\textwidth}
  \centering
  \includegraphics[width=\textwidth]{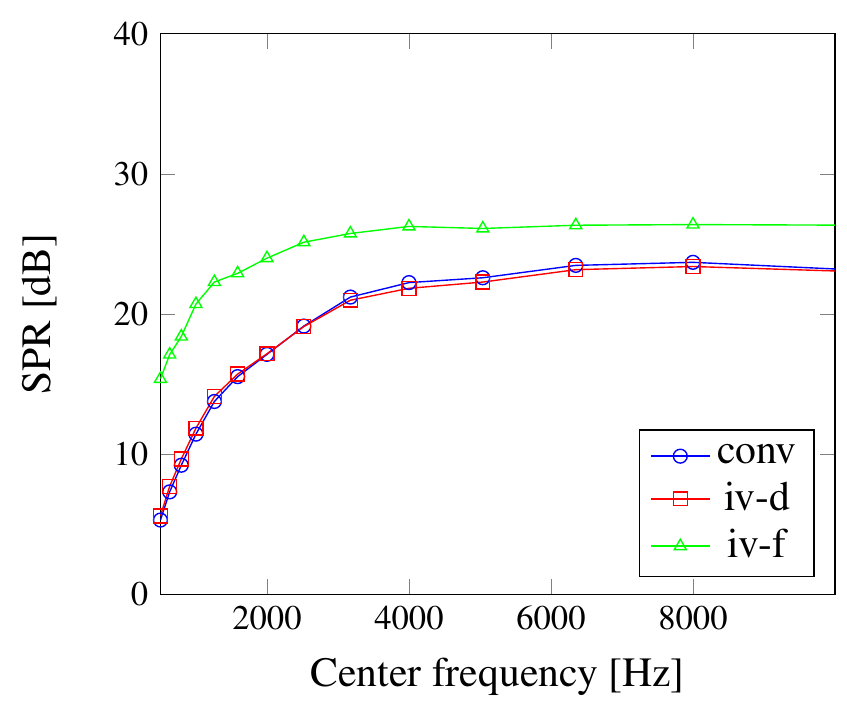}
  \caption{$\Delta_{\mathrm{noise}} = 20$dB}
\end{subfigure}
\begin{subfigure}{.32\textwidth}
  \centering
  \includegraphics[width=\textwidth]{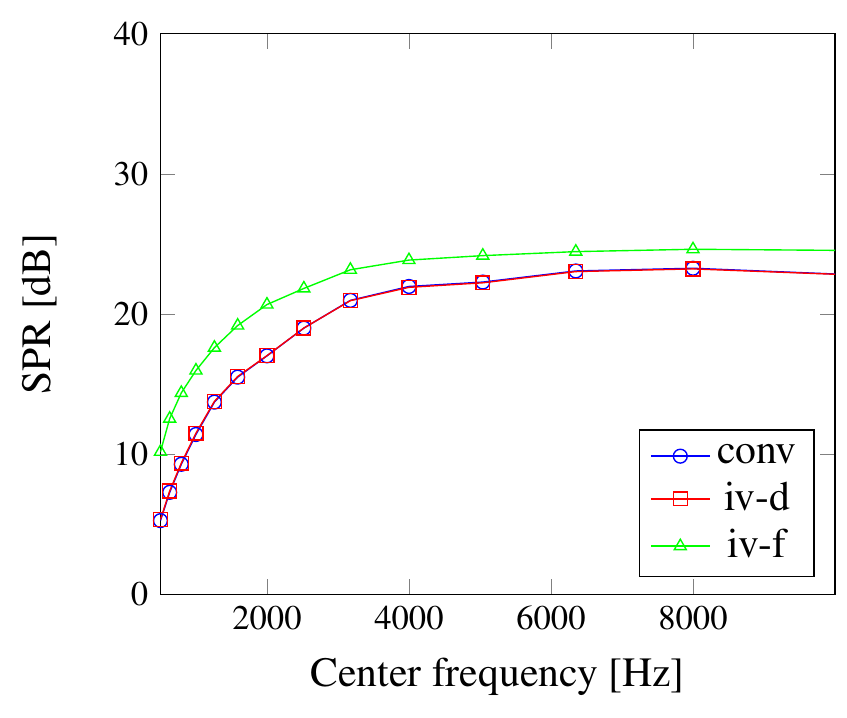}
  \caption{$\Delta_{\mathrm{noise}} = 10$dB}
\end{subfigure}
\begin{subfigure}{.32\textwidth}
  \centering
  \includegraphics[width=\textwidth]{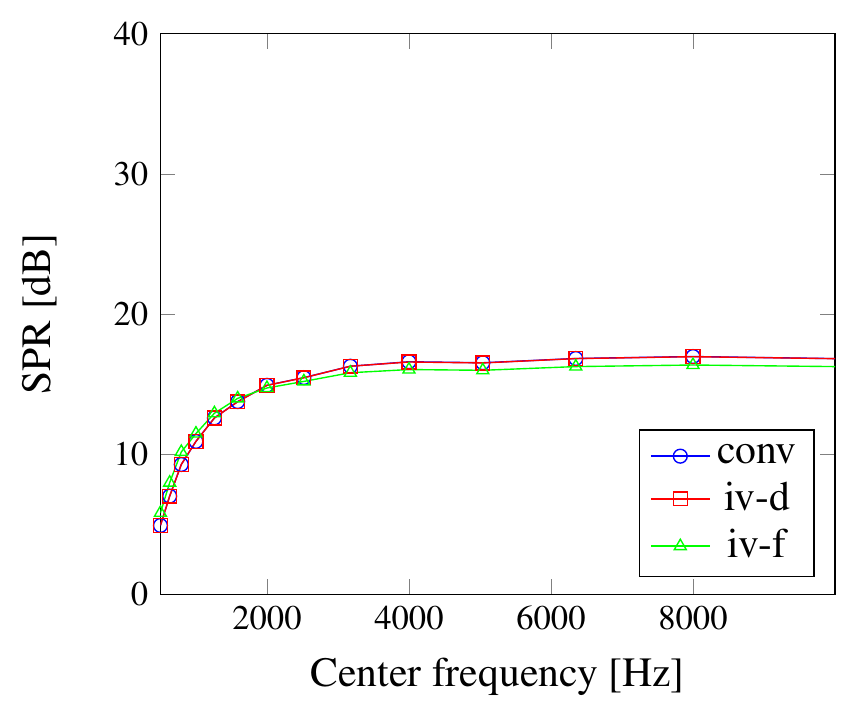}
  \caption{$\Delta_{\mathrm{noise}} = 0$dB}
\end{subfigure}%
\\
\vspace*{0.1cm}
\\
\begin{subfigure}{.32\textwidth}
  \centering
  \includegraphics[width=\textwidth]{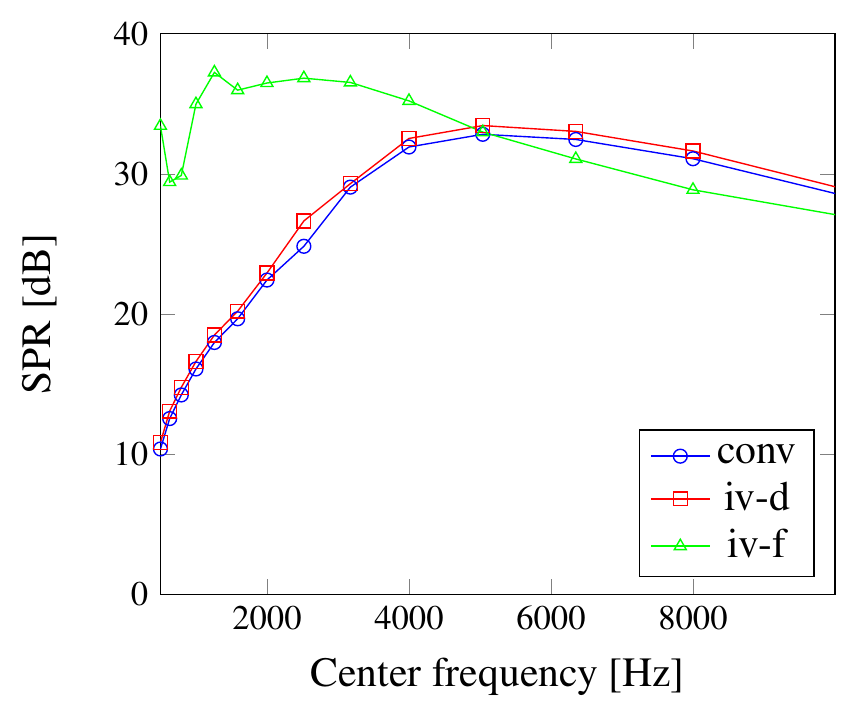}
  \caption{$\Delta_{\mathrm{noise}} = 20$dB}
\end{subfigure}
\begin{subfigure}{.32\textwidth}
  \centering
  \includegraphics[width=\textwidth]{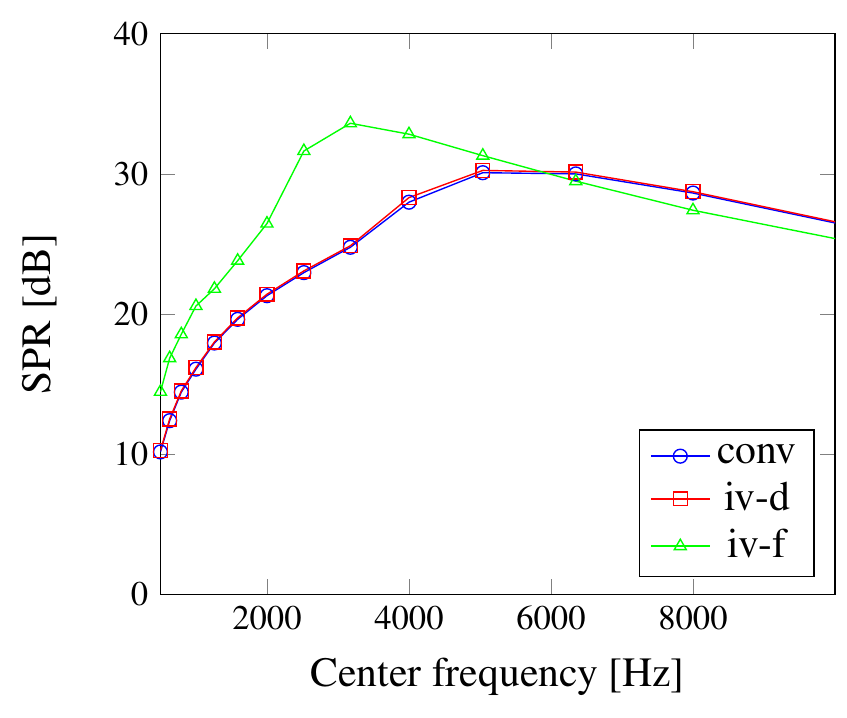}
  \caption{$\Delta_{\mathrm{noise}} = 10$dB}
\end{subfigure}
\begin{subfigure}{.32\textwidth}
  \centering
  \includegraphics[width=\textwidth]{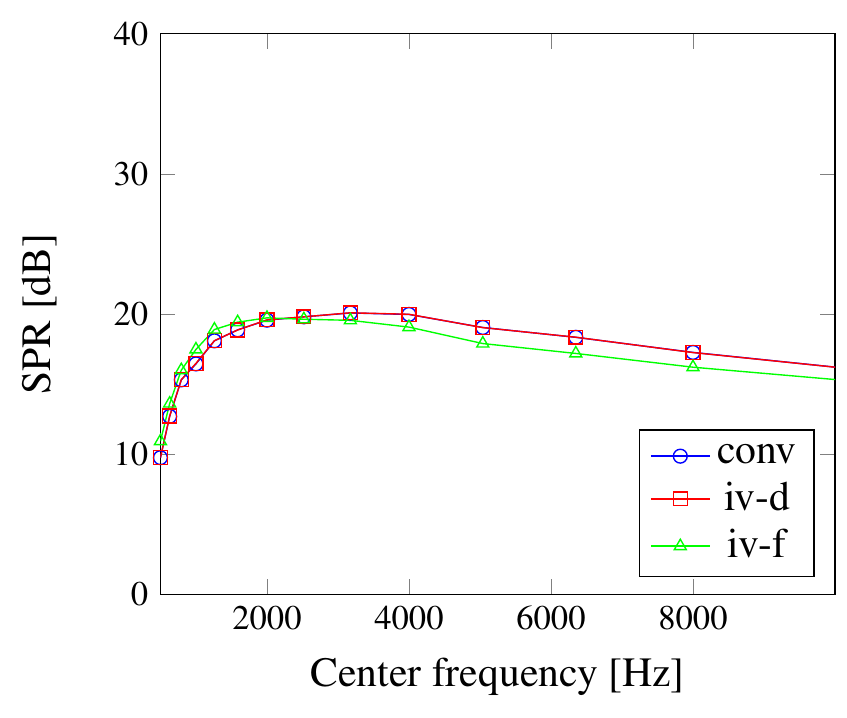}
  \caption{$\Delta_{\mathrm{noise}} = 0$dB}
\end{subfigure}%
\caption{SPR measure for beamforming (a)-(c) and regularized DAMAS-NNLS with discrepancy principle (d)-(f).}
\label{fig:spr-maps}
\end{figure}
\subsection{Effects on Resolution and Noise}
The measures on resolution \eqref{eq:res_measure}, SNR \eqref{eq:snr_measure} and SPR \eqref{eq:spr_measure} are applied to third-octave band sound pressure level (SPL) source maps, where the diagonal elements of the CSM are not used (diagonal removal). The center frequencies range from $500$Hz to $10$kHz. For almost all cases, the iv-d and the standard weighting results are very similar. This can be explained by the simple structure of the data set (cf. \eqref{eq:synth-samples}). By construction, the variances of the CSM components do not vary strongly and hence the iv-d weighting is very similar to the uniform weighting. The iv-f weighting on the other hand shows significant effects. 
\paragraph{Resolution}
The results regarding the resolution measure are shown in Fig.\ref{fig:res-maps}. The resolution of beamforming maps can be improved by the iv-f weighting. The improvements become smaller with increasing noise or increasing frequency. Similar effects can be observed for the regularized DAMAS-NNLS results. 
\paragraph{SNR}
The results regarding the SNR measure are shown in Fig.\ref{fig:snr-maps}. The iv-f weighting is able to improve the SNR on beamforming maps compared to the standard (or iv-d) weighting for the cases $\Delta_{\mathrm{noise}} \in \lbrace 20 \mathrm{dB}, 10 \mathrm{dB} \rbrace$. For $\Delta_{\mathrm{noise}} = 0$dB the SNR measure is very similar for all three weightings. For the DAMAS-NNLS results on the other hand, the iv-f weighting does not yield higher SNR measures. In the case of $\Delta_{\mathrm{noise}} = 20$dB the uniform weighting (as well as iv-d) has a strong peak of the SNR measure at low frequencies whereas the iv-f weighting has not. Those large SNR values are caused by the coarse resolution in that frequency range (cf. Fig\ref{fig:res-maps}). Hence, the resulting mainlobe dominates the sourcemap and the maximum sidelobe is pushed away from the map centre.
\paragraph{SPR}
The results regarding the SPR measure are shown in Fig.\ref{fig:spr-maps}. The iv-f beamforming maps have a higher SPR value for the cases $\Delta_{\mathrm{noise}} \in \lbrace 20 \mathrm{dB}, 10 \mathrm{dB} \rbrace$, whereas for the case with the strongest noise, all three weightings yield similar values. The effects of iv-f on the DAMAS-NNLS maps are quite similiar. For $\Delta_{\mathrm{noise}} \in \lbrace 20 \mathrm{dB}, 10 \mathrm{dB} \rbrace$ iv-f shows higher SPR values especially at lower frequencies. For all noise levels we observe that at higher frequencies the SPR of the iv-f weighted maps is slightly below the SPR of the other two weightings.
%%%%%%%%%%%%%%%%%%%%%%%%%%%%%%%%%%%%%%%%%%%%%%%%%%%%%%%%%%%%%
%%%%%%%%%%%%%%%%% Results- experimental %%%%%%%%%%%%%%%%%%%%%
%%%%%%%%%%%%%%%%%%%%%%%%%%%%%%%%%%%%%%%%%%%%%%%%%%%%%%%%%%%%%

\section{Results on Experimental Data}\label{sec:results-experimental}
Now we investigate how the previously discussed weighting choices affect the beamforming and DAMAS-NNLS source maps for real measurement data. We will work on a benchmark dataset of measurements at the cryogenic wind tunnel in Cologne (DNW-KKK) \cite{Ahlefeldt2013}, \cite{Bahr2017}. The test scenario considers a scaled DO-728 half-model at Mach numbers 0.15, 0.2, 0.25 and angles of attack 3$^\circ$, 5$^\circ$, 9$^\circ$. We consider measurements conducted at a static pressure of $1007 \ \mathrm{Pa}$, a static temperature of $286 \  \mathrm{K}$ and a Reynolds number of approximately $1.26 \cdot 10^6$. The benchmark dataset is available at the website of the technical university of Cottbus and is labeled 'DLR1' \cite{BenchmarkDLR1}.  \ \\ \ \\
We will illustrate the impact on the imaging result for the three weighting choices: conventional \eqref{eq:wc_conventional}, inverse (co)variance - diagonal \eqref{eq:wc_iv-d}  and inverse (co)variance - full \eqref{eq:wc_iv-f}. This will be done for the discussed methods beamforming \eqref{eq:wbf_def} and regularized DAMAS-NNLS \eqref{eq:damas_nnlsw_reg}. We will also discuss the choice of the covariance estimator and some statistical properties of the dataset. It should be emphasized that the results shown are primarily intended to illustrate the concept of weighted beamforming and DAMAS-NNLS on experimental data. As we consider only a specific dataset and a few selected source maps, the results should not be understood as a detailed comparative study of the different methods. 
\subsection{Covariance estimator}
For the iv-d weighting, sample variances are used since that approach does not rely on any additional assumptions on the data. For the iv-f weighting we employ the estimator $\mathbf{\Sigma}^{\text{est}}$ given by Eq. \eqref{eq:covcov_formula}. As mentioned beforehand, this estimation procedure assumes zero mean Gaussian signals and does neither guarantee regularity nor positive semi-definiteness. However, for the considered dataset the matrix $\mathbf{\Sigma}^{\text{est}}$ appears to be positive definite and regular.  Within the frequency range from $3000$ Hz to $9000$ Hz, the order of magnitude of the minimal eigenvalue $\lambda_{\mathrm{min}}$ ranges from $\mathcal{O}\left( 10^{-7} \right)$ to $\mathcal{O}\left( 10^{-3} \right)$. The condition number of the covariance estimator with respect to the Euclidean norm is given by the ratio between the largest and smallest eigenvalue i.e. $\text{cond}_2 \left( \mathbf{\Sigma}^{\text{est}}  \right) =  \frac{\lambda_{\mathrm{max}}}{\lambda_{\mathrm{min}}}$. The order of magnitude of the condition number for frequencies between $3000$ Hz and $9000$ Hz ranges from $\mathcal{O}\left( 10^{2} \right)$ to $\mathcal{O}\left( 10^{3} \right)$. The assumption of zero mean Gaussian signals is supported by the results of the next subsection. 
\subsection{Statistical properties}
We have seen that the choice of the weighting matrix $\mathbf{W}$ is often closely connected to statistical assumptions about the data (see e.g. Remark \ref{rem:whitening} and Assumption \ref{as:proper}). In the following we will examine some statistical quantities that may serve as indicators for the plausibility of those assumptions. The quantities are evaluated for the considered dataset within the frequency range from $3$kHz to $9$kHz. 

\paragraph{Zero mean pressure signals} \ \\
Assumption \ref{as:proper}.\ref{it:scnd_p} states that the pressure signal vector $\mathbf{p}$ has zero mean. As a measure of the deviation from this assumption we define
\begin{linenomath}\begin{equation*}
\varepsilon_{\mathrm{mean}} = \sqrt{\frac{\sum_{m=1}^M \abs{\mathbb{M}\left(p_j(\mathbf{x}_m) \right)}^2}{\sum_{m=1}^M \abs{\mathbb{M} \left(\abs{p_j(\mathbf{x}_m)} \right) }^2}} \ .
\end{equation*}\end{linenomath} 
This measure relates the norm of the mean pressure vector to the norm of the mean vector of absolute pressure values. Fig. \ref{fig:zeroMean} shows that $\varepsilon_{\mathrm{mean}}$ is below $0.04$ for all frequencies. Evaluating $\varepsilon_{\mathrm{mean}}$ on the synthetic dataset from the previous section yields maximum values in the range of $0.07$ ($\Delta_{\mathrm{noise}}=20$dB) respectively  $0.04$ ($\Delta_{\mathrm{noise}} \in \lbrace 10 \mathrm{dB}, 0 \mathrm{dB} \rbrace$).

\paragraph{Gaussian pressure signals}
Assumption \ref{as:proper}.\ref{it:fst_p} states that the pressure signal vector $\mathbf{p}$ is a vector-valued complex Gaussian random variable. This assumption was tested using the Anderson-Darling test for normality \cite{Anderson1952, Anderson1954} with a significance level of $5\%$. Since the Anderson-Darling test is designed for real valued random samples, real and imaginary part of the pressure samples are tested individually. For each frequency the acceptance rate $R_{\mathrm{accept}}$ is defined as the ratio of microphones for which the null hypothesis (sample is normally distributed) is accepted. Fig. \ref{fig:adTest} shows that the acceptance rate is greater than $0.8$ for almost all frequencies.  

\paragraph{Proper pressure signals}
Assumption \ref{as:proper}.\ref{it:thrd_p} states that the pseudo correlation matrix 
$\mathbb{E}\left( \mathbf{p} \mathbf{p}^{\top} \right)$ vanishes. For the estimated quantities that implies that the measured PCSM $\mathbf{C}^{\mathrm{ps}}$ \eqref{eq:pcsm_def} must be negligible compared to the measured CSM $\cobs$ \eqref{eq:csm_def}. Therefore we consider the ratio of their norms
\begin{linenomath}\begin{equation*}
R_{\mathrm{proper}} = \frac{\norm{\mathbf{C}^{\mathrm{ps}}}_F}{\norm{\cobs}_F} \ .
\end{equation*}\end{linenomath}
Fig. \ref{fig:properRatio} shows that the norm ratio for almost all frequencies lies between $0.2$ and $0.3$. 

\paragraph{White noise}
The weighting choice $\mathbf{W} \sim \mathbf{I}$ that is used for standard imaging methods implicitly imposes the assumption of white noise (see Remark \ref{rem:whitening}). To measure the deviation from this assumption we consider
\begin{linenomath}\begin{equation*}
\varepsilon_{\mathrm{white}} = \frac{ \min \limits_{a \in \mathbb{R}^{+}}\norm{\mathbf{\Sigma}^{\mathrm{est}} - a \mathbf{I}}_F}{\norm{\mathbf{\Sigma}^{\mathrm{est}}}_F} \ ,
\end{equation*}\end{linenomath}
where the covariance estimator $\mathbf{\Sigma}^{\mathrm{est}}$ is obtained by Formula \eqref{eq:covcov_formula}. Fig. \ref{fig:whiteNoiseGauss} shows that $\varepsilon_{\mathrm{white}}$ is greater than $0.6$ for all frequencies.

\paragraph{Assessment of statistical evaluations:}
Fig. \ref{fig:zeroMean} shows only very small deviations from the zero mean assumption for the pressure signals (Assumption \ref{as:proper}.\ref{it:scnd_p}). Since a large majority of the microphone signals passes the Anderson-Darling test \ref{fig:adTest}, the assumption of Gaussian pressure signals (Assumption \ref{as:proper}.\ref{it:fst_p}) is also supported. According to Fig. \ref{fig:properRatio}, the observed PCSM is not negligible compared to the CSM. Those results indicate that the assumption of proper pressure signals (Assumption \ref{as:proper}.\ref{it:thrd_p}) seems to be violated by the dataset. Similarly, the results in Fig. \ref{fig:whiteNoiseGauss} indicate that the assumption of white noise is violated.

\begin{figure}[h]
\begin{subfigure}{.49\textwidth}
\centering
\includegraphics[width = \textwidth]{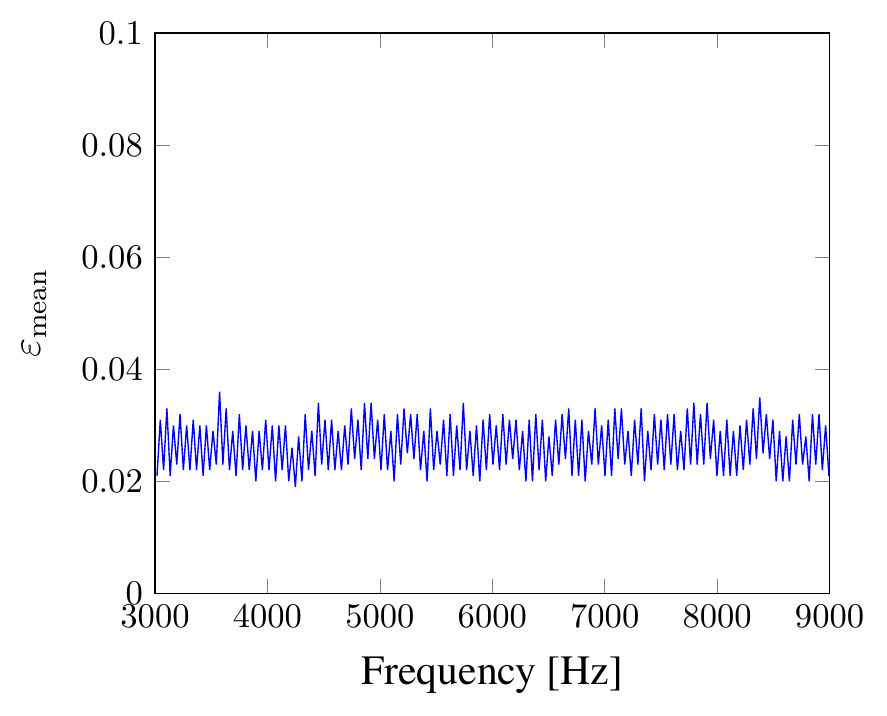}
\caption{Deviation from zero mean assumption. \label{fig:zeroMean}}
\end{subfigure}
\begin{subfigure}{.49\textwidth}
\centering
\includegraphics[width = \textwidth]{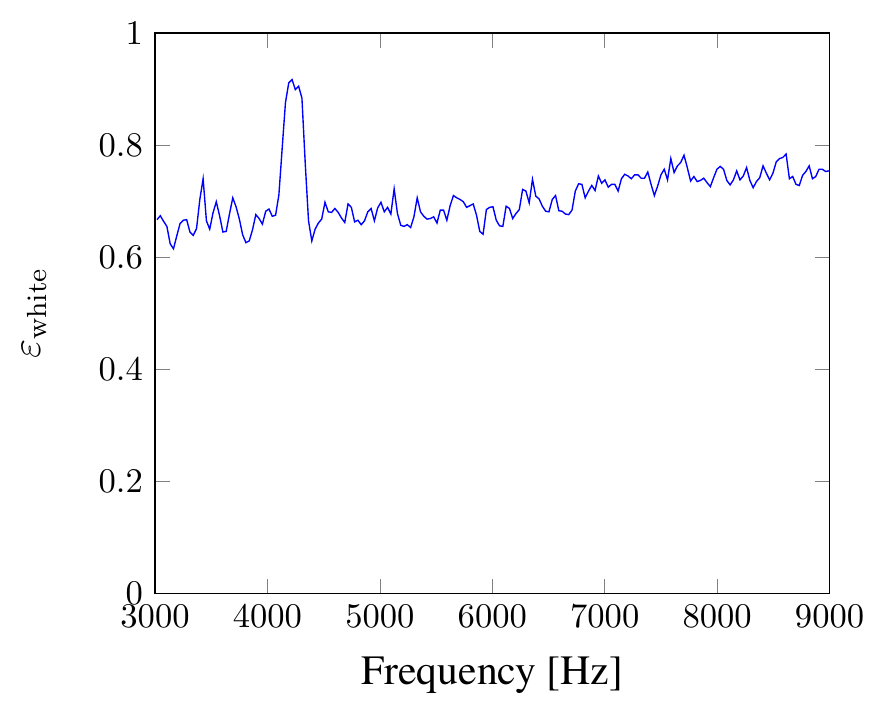}
\caption{Deviation from white noise assumption. \label{fig:whiteNoiseGauss}}
\end{subfigure}
\\
\vspace*{0.1cm}
\\
\begin{subfigure}{.49\textwidth}
\centering
\includegraphics[width = \textwidth]{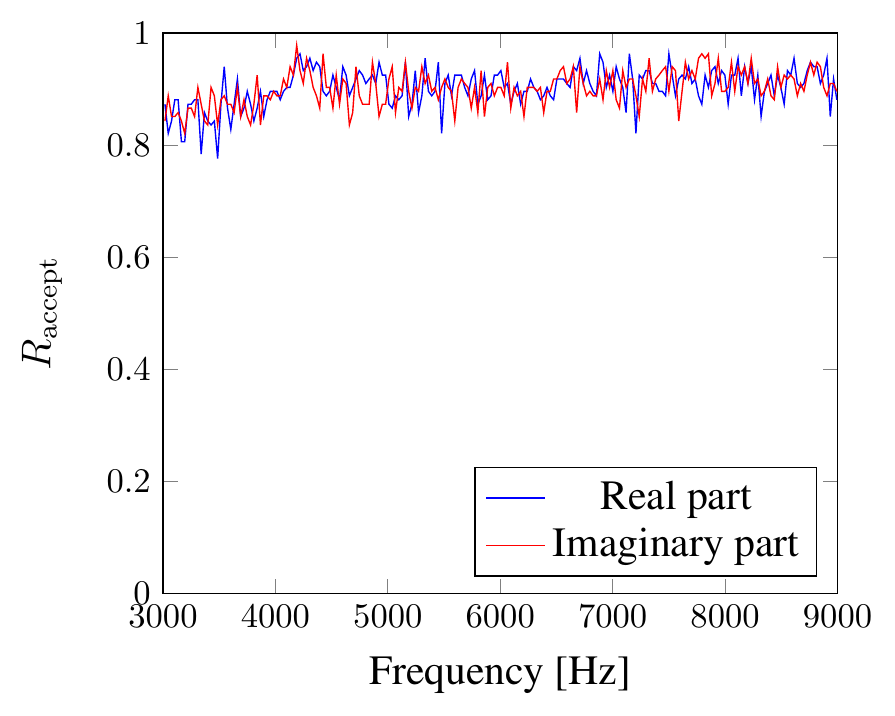}
\caption{Acceptance rate of Anderson-Darling test. \label{fig:adTest}}
\end{subfigure}
\begin{subfigure}{.49\textwidth}
\centering
\includegraphics[width = \textwidth]{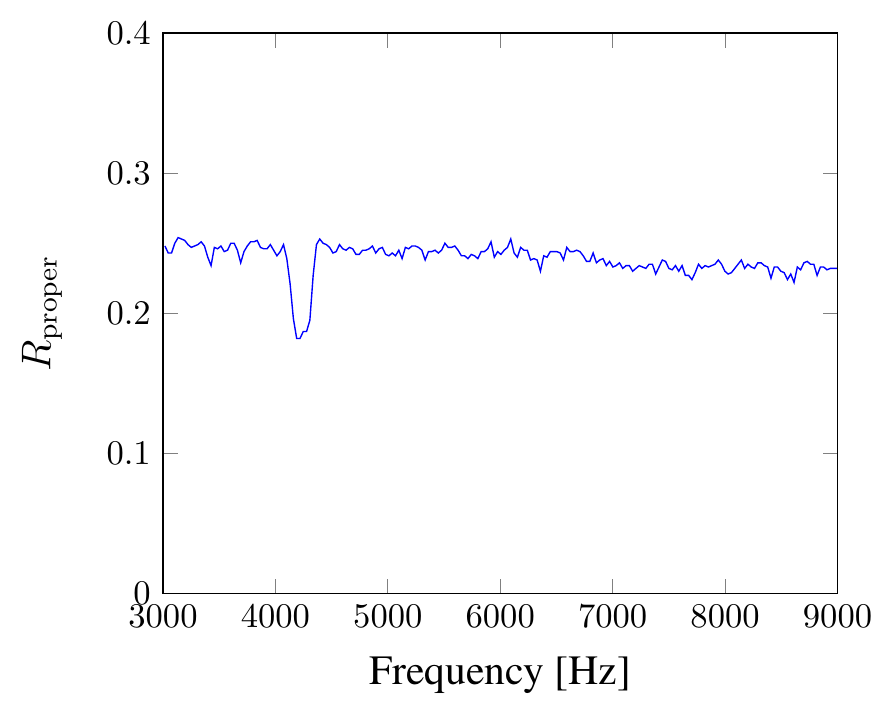}
\caption{Norm ratio of PCSM and CSM. \label{fig:properRatio}}
\end{subfigure}
\caption{Testing of statistical assumptions}
\label{fig:statistics}
\end{figure}

\subsection{Beamforming} 
Fig. \ref{fig:bf-maps} shows third-octave band beamforming source maps at three center frequencies (4000 Hz, 6000 Hz, 8000 Hz)  for iv-f, iv-d and Conventional Beamforming. The results are computed on a focus grid with $53 \times 73$ focus points and the autocorrelations were removed (diagonal removal) from the dataset. All source maps are adjusted to the local source maximum, respectively. The sources and source levels of all three used noise models appear similar for each frequency band, respectively. The maximum value of Conventional and iv-f beamforming is almost identical whereas the maximium of the iv-d beamformer deviates in the range of $0.5$ dB. The main sources are at the inner and outer slat region and the flap side edge. The source peaks are sharper at higher frequencies regardless of the Beamformer noise model. 
Both non-standard weightings show an increased dynamic range and resolution compared to the standard conventional weighting. The local peaks are sharper and the values in regions apart from the half-model, where no sources are expected, are lower. The results for iv-d and iv-f beamforming look very similar in many focus regions. At 8 kHz, sources at the inboard flap
leading edge are damped by the iv-d and the conventional beamformer such that they may not be identifiable anymore. Those minor sources are only identified by the iv-f beamformer.

\subsection{DAMAS-NNLS}

Fig. \ref{fig:damas-maps} shows third-octave band source maps at three center frequencies (4000 Hz, 6000 Hz, 8000 Hz) for the regularized DAMAS-NNLS solution with iv-f, iv-d and conventional weighting.
For each third-octave band and weighting type, the regularization parameter $\alpha$ was chosen according to the discrepancy principle \eqref{eq:discrepancy} at the center frequency with $\tau = 1.5$ and $\delta = \delta_{\mathbf{W}}^{\mathrm{rms}}$. Again the source maps are adjusted to their local source maximum.\ \\ \ \\
Both non-standard weighting approaches show an improved damping of noise effects at 4 kHz. The iv-f source map appears much cleaner than the source map for conventional weighting. At 8 kHz, the iv-f result shows a bit more noise effects apart from the wing than the other two weighting choices. The individual maxima of each map differ from each other in a range of maximum $2$ dB. For all frequencies the iv-f result has the highest resolution (i.e. the local peaks are the sharpest).  For each third-octave bands, the regularization parameter increases from top to bottom i.e. iv-d requires less regularization than conventional weighting and iv-f requires less regularization than iv-d. This observation is in accordance with the theory from Section \ref{sec:varopt} which states that iv-f has the lowest root mean squared noise level among all possible weightings. 

\begin{figure}[h]
\begin{subfigure}{.32\textwidth}
  \centering
  \includegraphics[width=\textwidth]{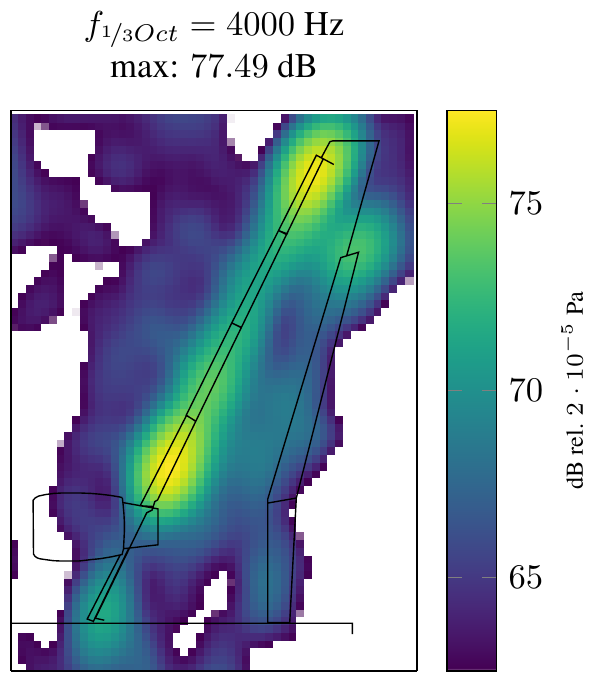}
  \caption{}
\end{subfigure}
\begin{subfigure}{.32\textwidth}
  \centering
  \includegraphics[width=\textwidth]{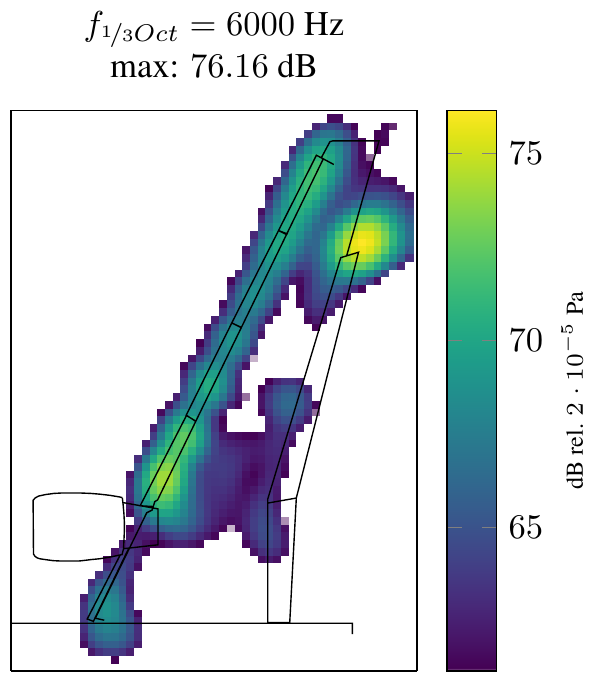}
  \caption{}
\end{subfigure}
\begin{subfigure}{.32\textwidth}
  \centering
  \includegraphics[width=\textwidth]{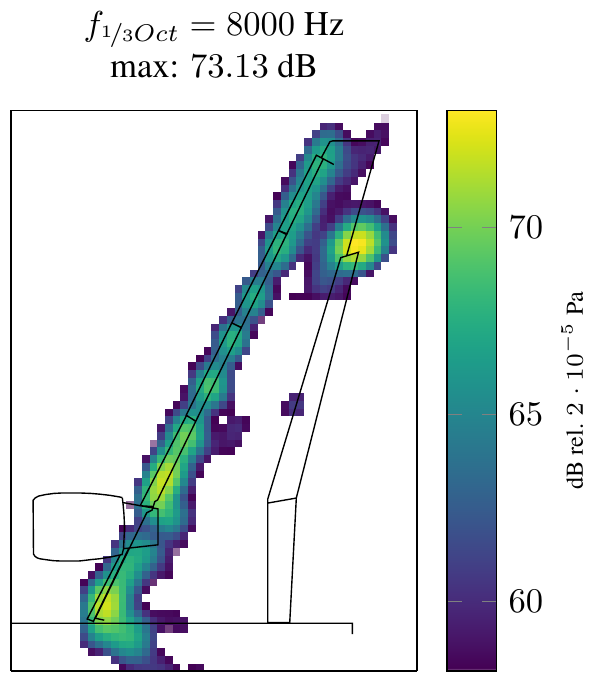}
  \caption{}
\end{subfigure}%
\\
\vspace*{0.1cm}
\\
\begin{subfigure}{.32\textwidth}
  \centering
  \includegraphics[width=\textwidth]{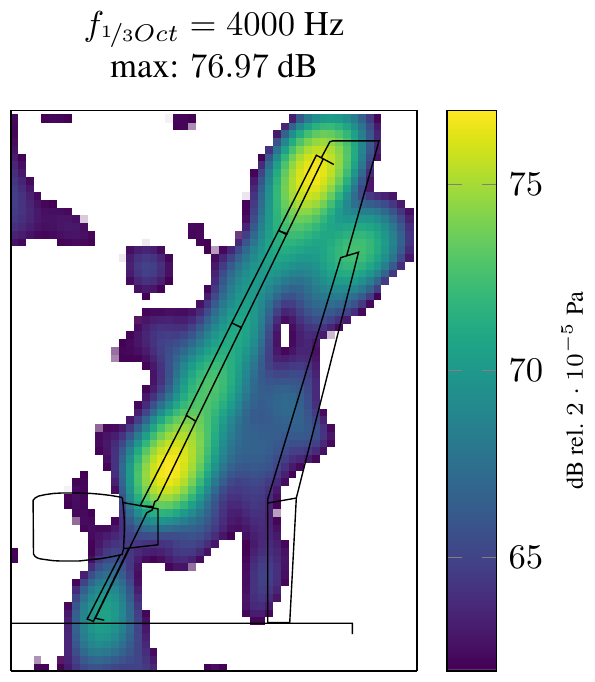}
  \caption{}
\end{subfigure}
\begin{subfigure}{.32\textwidth}
  \centering
  \includegraphics[width=\textwidth]{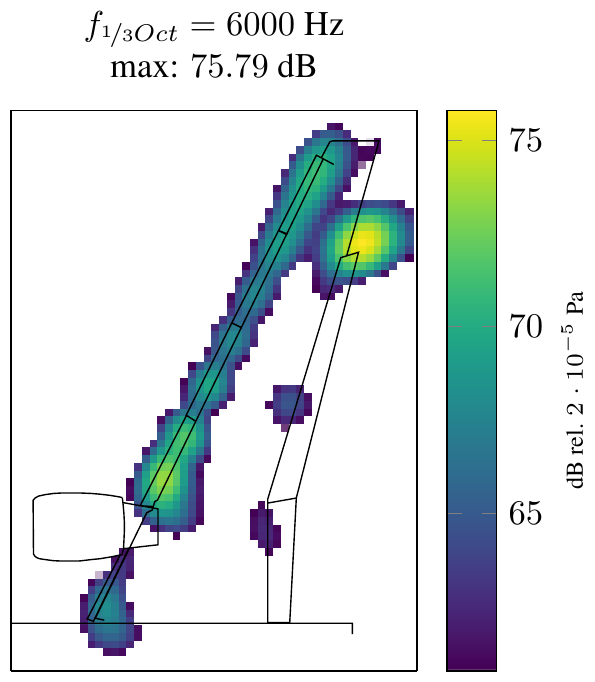}
  \caption{}
\end{subfigure}
\begin{subfigure}{.32\textwidth}
  \centering
  \includegraphics[width=\textwidth]{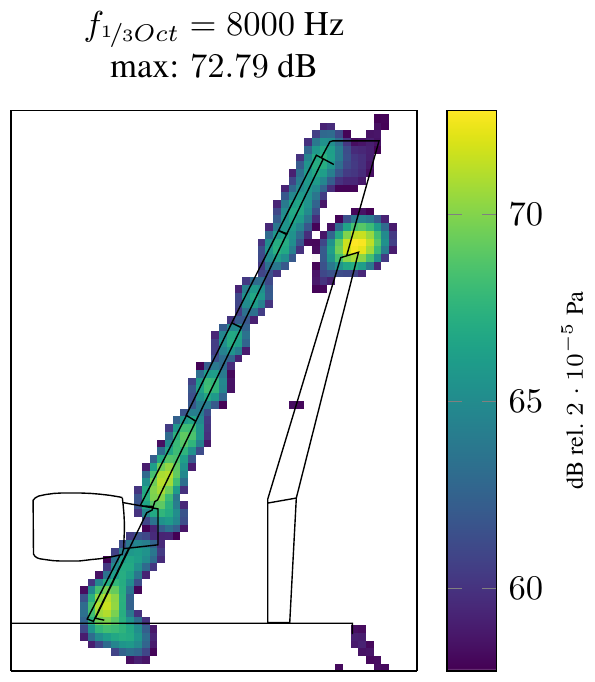}
  \caption{}
\end{subfigure}%
\\
\vspace*{0.1cm}
\\
\begin{subfigure}{.32\textwidth}
  \centering
  \includegraphics[width=\textwidth]{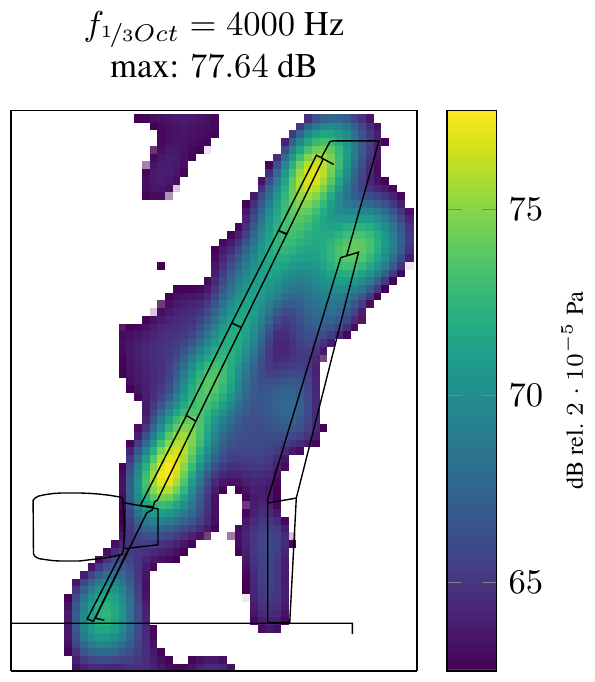}
  \caption{}
\end{subfigure}
\begin{subfigure}{.32\textwidth}
  \centering
  \includegraphics[width=\textwidth]{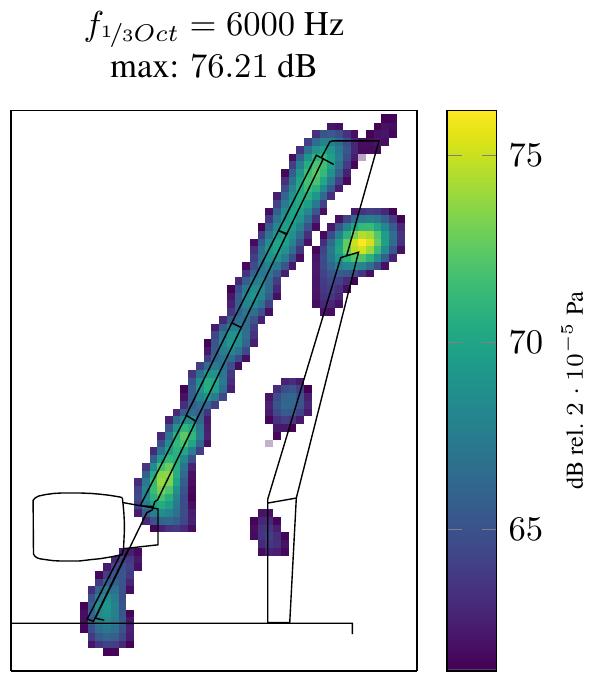}
  \caption{}
\end{subfigure}
\begin{subfigure}{.32\textwidth}
  \centering
  \includegraphics[width=\textwidth]{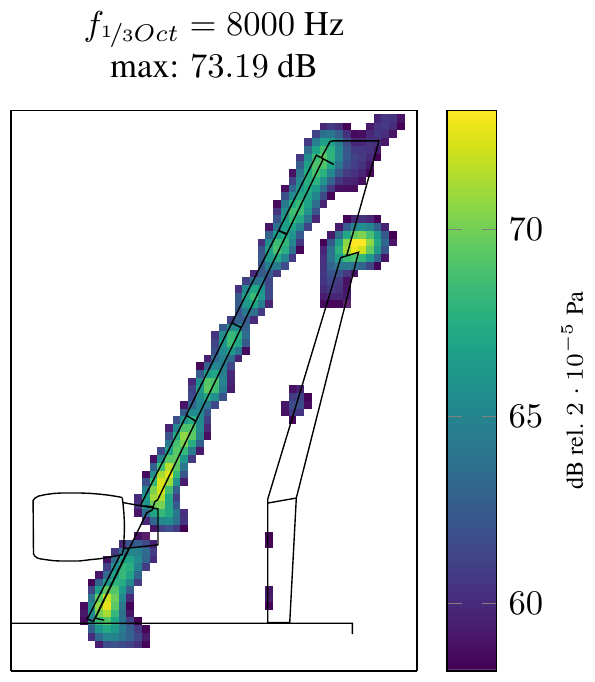}
  \caption{}
\end{subfigure}%
\caption{Third-octave band beamforming source maps for Conventional Beamforming (a)-(c) iv-d beamforming (d)-(f) and iv-f beamforming (g)-(i). Dynamic range: $15$ dB, Mach number: 0.15, angle of attack: 3$^\circ$.}
\label{fig:bf-maps}
\end{figure}

\begin{figure}[h]
\begin{subfigure}{.32\textwidth}
  \centering
  \includegraphics[width=\textwidth]{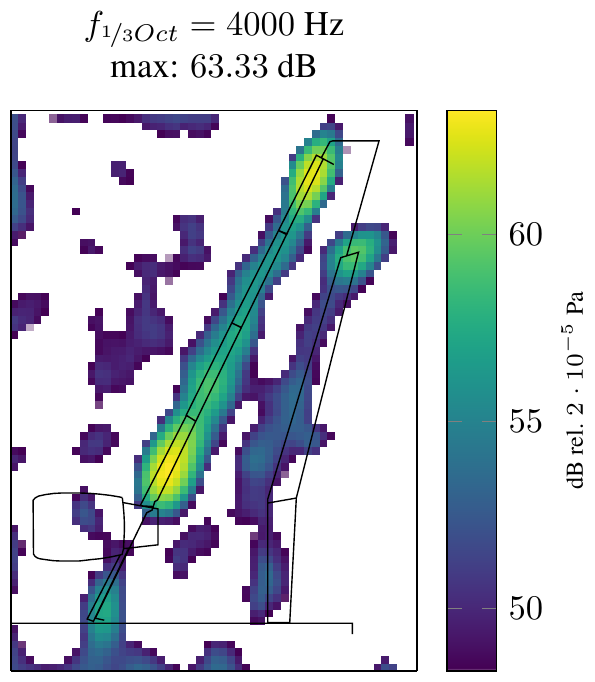}
  \caption{$\alpha \approx 63.57$}
\end{subfigure}
\begin{subfigure}{.32\textwidth}
  \centering
  \includegraphics[width=\textwidth]{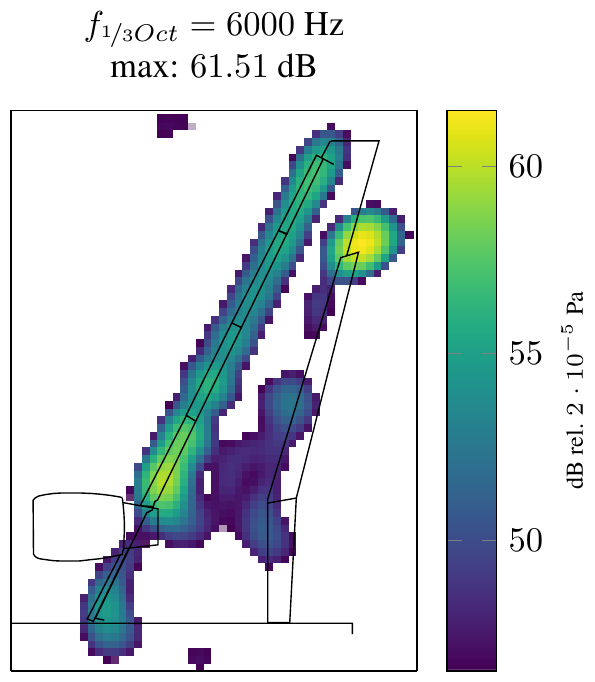}
   \caption{$\alpha \approx 77.37$}
\end{subfigure}
\begin{subfigure}{.32\textwidth}
  \centering
  \includegraphics[width=\textwidth]{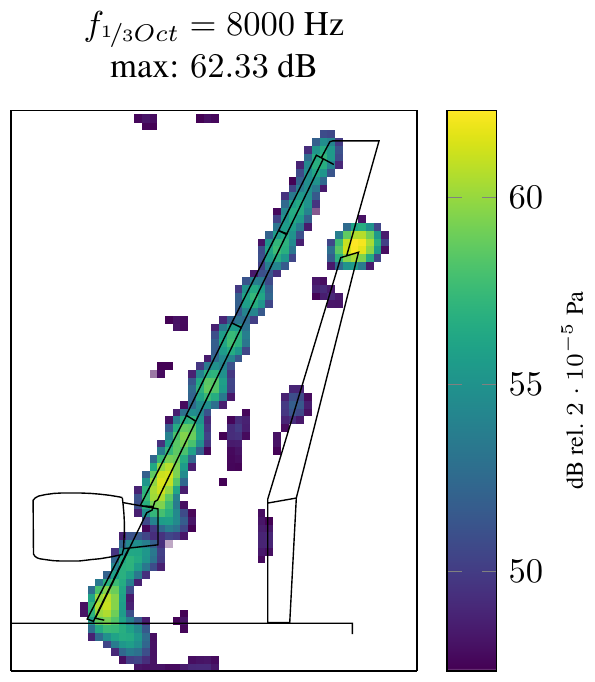}
  \caption{$\alpha \approx 26.41$}
\end{subfigure}%
\\
\vspace*{0.1cm}
\\
\begin{subfigure}{.32\textwidth}
  \centering
  \includegraphics[width=\textwidth]{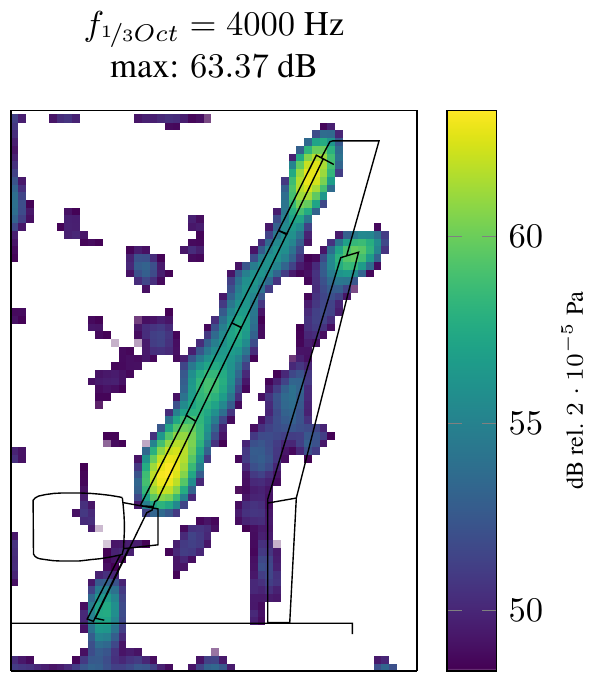}
  \caption{$\alpha \approx 51.17$}
\end{subfigure}
\begin{subfigure}{.32\textwidth}
  \centering
  \includegraphics[width=\textwidth]{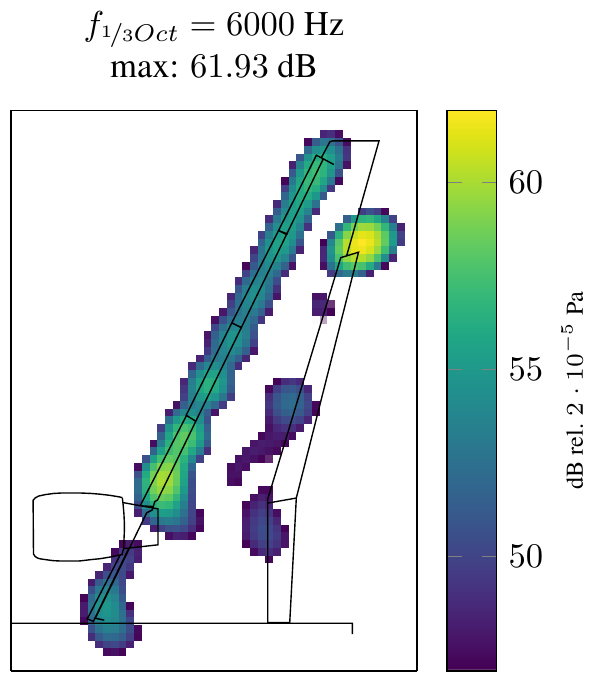}
  \caption{$\alpha \approx 63.08$}
\end{subfigure}
\begin{subfigure}{.32\textwidth}
  \centering
  \includegraphics[width=\textwidth]{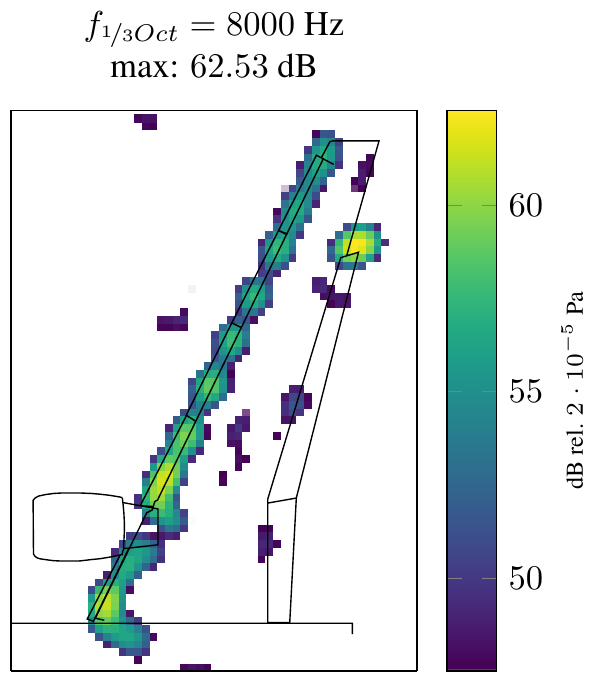}
  \caption{$\alpha \approx 14.85$}
\end{subfigure}%
\\
\vspace*{0.1cm}
\\
\begin{subfigure}{.32\textwidth}
  \centering
  \includegraphics[width=\textwidth]{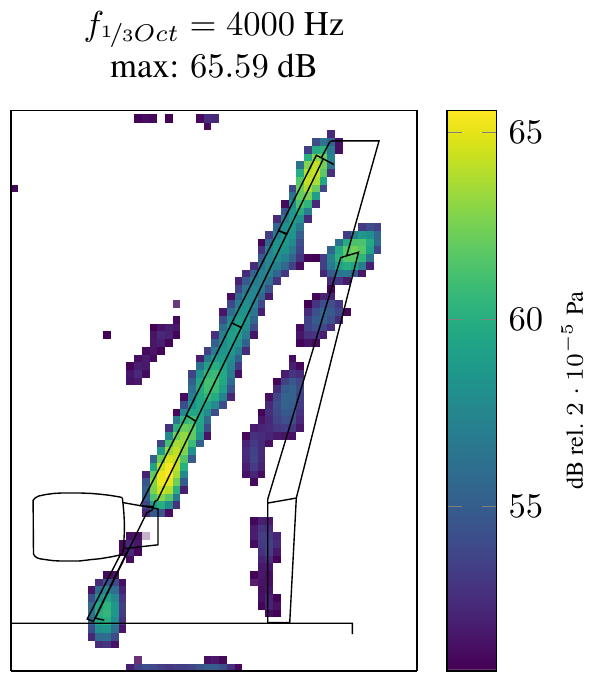}
  \caption{$\alpha \approx 10.54$}
\end{subfigure}
\begin{subfigure}{.32\textwidth}
  \centering
  \includegraphics[width=\textwidth]{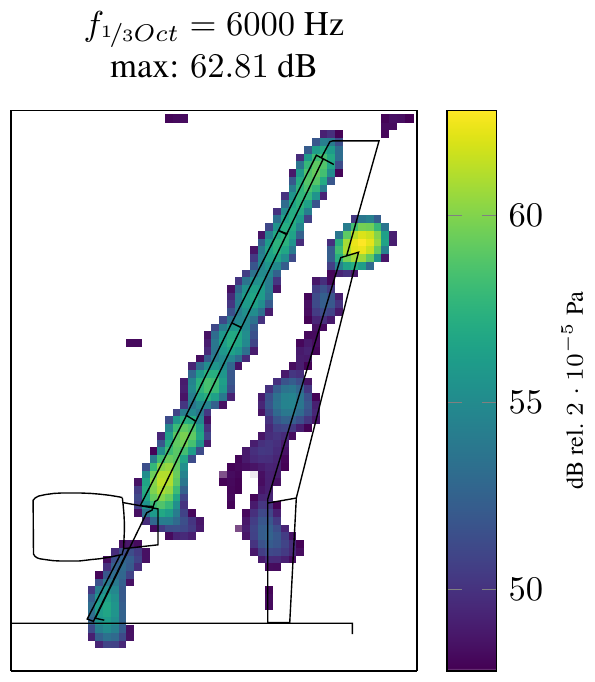}
  \caption{$\alpha \approx 21.60$}
\end{subfigure}
\begin{subfigure}{.32\textwidth}
  \centering
  \includegraphics[width=\textwidth]{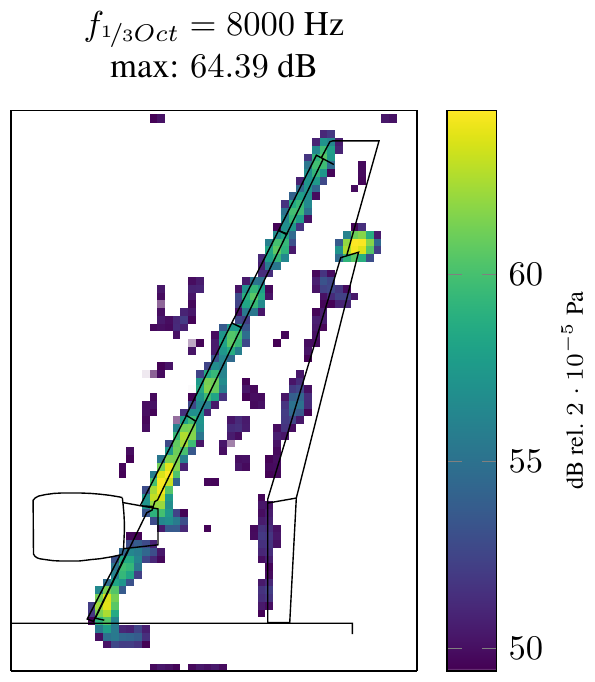}
  \caption{$\alpha \approx 3.78$}
\end{subfigure}%
\caption{Third-octave band, regularized DAMAS-NNLS source maps for conventional weighting (a)-(c) iv-d weighting (d)-(f) and iv-f weighting (g)-(i). Dynamic range: $15$ dB, Mach number: 0.15, angle of attack: 3$^\circ$.}
\label{fig:damas-maps}
\end{figure}

\FloatBarrier

%%%%%%%%%%%%%%%%%%%%%%%%%%%%%%%%%%%%%%%%%%%%%%%%%%%%%%%%%%%%%
%%%%%%%%%%%%%%%%%%%%% Conclusions %%%%%%%%%%%%%%%%%%%%%%%%%%%
%%%%%%%%%%%%%%%%%%%%%%%%%%%%%%%%%%%%%%%%%%%%%%%%%%%%%%%%%%%%%

\section{Conclusion}\label{sec:conclusion}
We examined a model for the measurement process that provides a representation of the observed cross correlations as the sum of acoustic correlations, hydrodynamic correlations and zero mean noise. The noise term can possess an arbitrary covariance structure in general. Under a Gaussian assumption on the pressure vector, all noise covariances can be estimated efficiently at least if the estimated covariance matrix is regular and positive definite. The generic noise model motivated a modified distance measure, parameterized by a weighting matrix which may be chosen as the full noise covariance matrix (iv-f) or its diagonal part (iv-d). Each weighted distance measure defines a beamforming and DAMAS-NNLS method by replacing the standard Euclidean distance by the weighted distance in the characterizing minimization problem. This led to a whole class of source localization techniques containing many well-known methods (CBF, RAB, Capon's method, shading). In a theoretical analysis we showed that among all weighting choices, the iv-f weighted beamformer has the lowest variance. Furthermore we demonstrated that Capon's method (a.k.a. minimum variance method) is not equivalent to the iv-f beamformer. If the pressure signals are Gaussian with zero mean and proper, Capon's method yields a good approximation of the iv-f beamformer. However, the statistical findings indicate that the assumption of proper pressure signals does not seem to be valid for the considered experimental dataset. The application of the iv-f weighting on sythetic data yields improved resolution for beamforming and regularized DAMAS-NNLS. For beamforming maps, the iv-f weighting has also a positive effect on the SNR. On the experimental dataset the weightings showed improvements of the source map quality for beamforming as well as for regularized DAMAS-NNLS, especially at lower frequencies. Both data dependent weighting choices, iv-d and iv-f were able to increase the resolution and reduce noise effects in the source maps. Since iv-d has the same computational order of complexity as Conventional beamforming it may be very attractive in terms of efficiency.  
\ \\ \ \\If we consider the Mahalanobis distance (iv-f weighting) as the natural distance measure in a measurement environment with additive random noise, the standard methods implicitly impose that the noise is white. Whenever this white noise assumption is violated, the quality of source maps obtained by standard methods can suffer. We clearly showed in the theoretical analysis and the results section that source localization methods can benefit from fourth order moments of the pressure data. Even if the covariance structure of the noise is not directly incorporated into the method itself it can be used to test how strongly the white noise assumption is violated by the dataset. This can serve as an indicator of uncertainty of standard weighted source maps.  
\FloatBarrier
\bibliography{references}

\end{document}